\documentclass[11pt]{article}
\usepackage{graphicx}
\usepackage{amsfonts}
\usepackage{amsmath}
\usepackage{amssymb}
\usepackage{algorithm}
\usepackage{algorithmic}
\usepackage{url}
\usepackage{fancyhdr}
\usepackage{indentfirst}
\usepackage{enumerate}
\usepackage{dsfont,footmisc}
\usepackage{comment}
\usepackage{graphicx}
\usepackage{caption}   
\usepackage{tikz}
\usetikzlibrary{arrows.meta, automata, positioning, calc}
\usepackage{booktabs}
\usepackage{subcaption}

\usepackage{stmaryrd}
\usepackage{bbold}
\usepackage{mathrsfs}
\usepackage{graphicx}
\usepackage[export]{adjustbox}
\allowdisplaybreaks[4]

\usepackage{setspace}
\usepackage[misc]{ifsym}

\usepackage{amsthm}
\usepackage{color}
\usepackage{natbib}
\usepackage[british]{babel}
\usepackage[colorlinks=true,citecolor=blue]{hyperref}
\usepackage[right,pagewise,displaymath, mathlines]{lineno}

\usepackage{tikz}

\usetikzlibrary{decorations.pathreplacing}
\usepackage{tikz-qtree}

\usetikzlibrary{arrows.meta,positioning}

\def\d{\mathrm{d}}

\newcommand{\VaR}{\mathrm{VaR}}

\newcommand{\ES}{\mathrm{ES}}
\newcommand{\E}{\mathbb{E}}
\newcommand{\R}{\mathbb{R}}

\newcommand{\p}{\mathbb{P}}

\newcommand{\id}{\mathds{1}}

\newcommand{\X}{\mathcal X}

\renewcommand{\le}{\leqslant}
\renewcommand{\geq}{\geqslant}
\renewcommand{\leq}{\leqslant}
\renewcommand{\epsilon}{\varepsilon}

\theoremstyle{plain}
\newtheorem{theorem}{Theorem}
\newtheorem{lemma}{Lemma}
\newtheorem{proposition}{Proposition}

\theoremstyle{definition}
\newtheorem{definition}{Definition}
\newtheorem{example}{Example}

\newtheorem{assumption}{Assumption}

\theoremstyle{remark}
\newtheorem{remark}{Remark}

\theoremstyle{definition}

\renewcommand{\cite}{\citet}
\renewcommand{\cdots}{\dots}

\begin{document} 
	
\title{
Risk-sensitive Reinforcement Learning Based on\\ Convex Scoring Functions} 

\author{
    Shanyu Han\thanks{\scriptsize School of Mathematical Sciences, Peking University, Beijing, China. \Letter~\texttt{hsy.1123@pku.edu.cn
}}
    \and
	Yang Liu\thanks{\scriptsize Corresponding Author. School of Science and Engineering, The Chinese University of Hong Kong (Shenzhen), Shenzhen, China. \Letter~\texttt{yangliu16@cuhk.edu.cn}}
    \and 
    Xiang Yu\thanks{\scriptsize Department of Applied Mathematics, The Hong Kong Polytechnic University, Kowloon, Hong Kong, China. \Letter~\texttt{xiang.yu@polyu.edu.hk}}
}

\date{\vspace{-0.5in}}

\maketitle
\begin{abstract}
{We propose a reinforcement learning (RL) framework under a broad class of risk objectives, characterized by convex scoring functions.  This class covers many common risk measures, such as variance, Expected Shortfall, entropic Value-at-Risk, and mean-risk utility. To resolve the time-inconsistency issue, we consider an augmented state space and an auxiliary variable and recast the problem as a two-state optimization problem. We propose a customized Actor-Critic algorithm and establish some theoretical approximation guarantees. A key theoretical contribution is that our results do not require the Markov decision process to be continuous. Additionally, we propose an auxiliary variable sampling method inspired by the alternating minimization algorithm, which is convergent under certain conditions. We validate our approach in simulation experiments with a financial application in statistical arbitrage trading, demonstrating the effectiveness of the algorithm.
}

\vspace{0.1in}
\noindent \textbf{Keywords}: Risk-sensitive reinforcement learning,  Actor-Critic algorithm, convex scoring function, augmented state process, two-stage optimization, error analysis
\end{abstract}

\vspace{0.3in}
\section{Introduction}

In the new trend of financial innovations, reinforcement learning (RL) has gained significant attention in wide financial decision-making when the environment is unknown (\cite{Hambly2023}). 
RL features the learning process when the agent perceives and interacts with the unknown environment and updates the decisions based on the trial-and-error procedure (\cite{Garcia2015}). 
In the standard RL algorithms, one considers the linear accumulation of observed rewards in the optimal decision making: 
\begin{equation}\label{prob:neutral}
\begin{aligned}
&\quad \inf_{\pi \in \Pi} \E \left[  \sum_{t = 0}^T c(S_t, A_t,S_{t+1})  \right],
\end{aligned}
\end{equation}
where $\pi \in \Pi$ is a policy and $c$ is a cost function. Note that maximizing the reward is equivalent to minimizing the cost. This traditional formulation actually shows a risk-neutral attitude towards the cost. However, from a practical point of view, a decision maker often exhibits strong risk avoidance rather than risk neutrality throughout the decision process, such as in autonomous driving and clinical treatment planning, because some extreme scenarios during the learning process may cause a significant or fatal loss. As a remedy, risk-sensitive RL algorithms have caught a lot of attention; see, e.g., \cite{Garcia2015,CGJP,FYCWX,FYCW,zhang2021mean,JPWT2022,bauerle2022markov,ni2022policy,du2022provably,wang2024reductions,Jia} for some recent developments. Typically, some types of risk preferences are incorporated into the observed rewards to depict the agent's risk aversion attitude on the learned policy. The goal of the present paper is to contribute to the topic of risk-sensitive RL methods by featuring the general risk measures associated with convex scoring functions.  



Scoring functions, a fundamental tool in statistics and quantitative risk management, offer a robust framework for evaluating forecasting methodologies. This makes them a natural choice for the risk preference to integrate into our RL approach. Specifically, they serve as metrics to assess and compare competing point forecasters or forecasting procedures. Some common examples include the absolute error and the squared error (see \cite{FZ16} and \cite{Brehmer2019}). 
In our formulation, the decision-maker aims to minimize the scoring function; see also the related study in  \cite{frongillo2021elicitation}. In the statistical contexts, the utilization of the scoring function is to align the predicted value closely to the true value, thereby the scoring function is often regarded as a loss function (see \cite{akbari2021does} and \cite{tyralis2025transformations}).

The concept of (classical) scoring functions is closely tied to the statistical idea of elicitability, which refers to the ability to compare competing forecasts using scoring functions (see \cite{gneiting2011making}, \cite{Ziegel16}, and the discussion in Section \ref{sec:riskmeasure}). This connection has profound implications for risk management theory (e.g., \cite{FS16}, \cite{EMWW21}, and \cite{FLWW2024}). A key application of scoring functions lies in the study of elicitable risk measures, where scoring functions are employed under specific conditions to analyze their properties (\cite{bellini2015elicitable}). Scoring functions also naturally align with various risk functionals, including generalized deviation measures (\cite{Rockafellar2006}) and sensitivity measures (\cite{fissler2023sensitivity}). Moreover, they serve as a bridge to the development of robust risk-averse optimization frameworks, as demonstrated in \cite{FLW24} and \cite{RMM25}. Recently, the distributional misspecification in risk analysis and forecast evaluation are addressed in different aspects and various robustness techniques are adopted; 
see \cite{Brehmer2019}, \cite{chen2022ordering}, \cite{Miao2024}, and \cite{BLLW25} for recent progress.


In the present paper, we consider a class of risk-sensitive RL problems by employing convex scoring functions on the expected cumulative costs within the MDP framework. The convex scoring function $f$ aims to quantify the discrepancy $f(y, \upsilon)$ between a cumulative cost $y$ and an auxiliary variable $\upsilon.$ The formal definition of convex scoring functions is given in Section \ref{sec:setup}. Overall, we are interested in devising some RL algorithms for the following risk-sensitive MDP problem when the model parameters are unknown: 
\begin{equation}
\label{prob:main}
\begin{aligned}
&\quad \inf_{\pi \in \Pi} \inf_{\upsilon\in\mathbb{R}} h\left(\E \left[  f\left( \sum_{t = 0}^T c(S_t, A_t,S_{t+1}) ,\upsilon\right) \right],\upsilon\right)
\end{aligned}
\end{equation} for some transform function $h.$
This framework covers important examples such as Expected Shortfall (ES), variance, and mean-risk utilities. In contrast to prior approaches that rely primarily on monetary or coherent risk measures, the proposed framework admits a broader class of risk functionals defined via convex scoring functions. This generality enables the treatment of both standard and non-standard measures. The inclusion of mean-variance utility demonstrates the framework’s applicability to a wide range of portfolio optimization problems. The risk functionals induced by convex scoring functions are closely related to optimized certainty equivalent (OCE) risk measures (see \cite{ben2007old,wang2024reductions}), which can be embedded into our framework. Beyond the OCE class, the framework can also accommodate other risk measures such as the entropic Value-at-Risk (EVaR) and utility-based risk measures; see Section \ref{sec:method} for details.

However, due to the time-non-separable nature in \eqref{prob:main}, it is well known that the time-inconsistency issue arises so that the standard dynamic programming principle cannot be employed directly in devising the algorithm. Our technical and practical contributions are summarized below.     

Firstly, we reformulate the risk-sensitive RL problem as a two-stage optimization problem, for which we consider the augmentation of the state space for the inner dynamic optimization problem. As a result, the inner optimization problem on the enlarged state space becomes time-consistent, and a recursive Bellman equation can be utilized for policy evaluation. Notably, our framework does not require continuity assumptions on the MDP, and we can show that the optimal value function can be effectively approximated by a sequence of policies (Theorem \ref{t_Bellman}). 

Secondly, we develop an Actor-Critic (AC) algorithm to optimize the policy and value function using neural networks. The Critic-step ensures that the value function approximates the true value uniformly, while the Actor-step updates the policy based on the policy gradient. Theoretical results demonstrate the approximation of the algorithm and provide bounds on the approximation error for both the value function (Theorem \ref{t1}) and the optimal auxiliary variable $\upsilon$ (Theorem \ref{t2}).

Thirdly, for the outer optimization over the auxiliary variable, we contribute to the selection and sampling of the optimal auxiliary variable to enhance computational efficiency. The selection step is conducted through stochastic gradient descent (SGD), minimizing a composite objective involving the value function and the scoring function. Theoretical guarantees ensure the approximation of this approach, and the distance between the selected $\upsilon$ and the true minimizers is rigorously analyzed (Theorem \ref{t2}). The efficient sampling of $\upsilon$ is conducted by an alternating minimization algorithm, combining critic and actor updates with adaptive sampling. This method ensures convergence under mild conditions (Theorem \ref{thm:limit}).

Finally, we apply the proposed AC algorithm to a statistical arbitrage problem in finance, where an RL agent manages trading strategies under risk constraints. We compare policies trained with different risk measures, such as ES and variance. Simulation experiments demonstrate the advantages of the proposed approach in minimizing risk-sensitive objectives compared with other baseline RL algorithms. 

The rest of the paper is structured as follows. Section \ref{sec:setup} provides preliminaries of the model setup, risk measures, and scoring functions. Section \ref{sec:model} formulates the risk-sensitive RL under the static risk measures, where we introduce the augmented state process to recast the problem as a time-consistent RL problem. Section \ref{sec:method} presents the theoretical foundations and the methodology of the proposed Actor-Critic RL algorithm, including the Bellman equation and the approximation of the proposed value function. Particularly, Sections \ref{sec:select} and \ref{sec:sample} elaborate on the selection and sampling of the optimal auxiliary variable $v$ in the RL algorithm. Section \ref{sec:app} applies the proposed methodology to a statistical arbitrage problem, comparing the performance of various risk measures through simulation experiments. Finally, Section \ref{sec:proof} collects all theoretical proofs. 

\section{Preliminaries}\label{sec:setup}

\subsection{Markov Decision Process}
In the RL setting, the agent learns an optimal policy through interactions with the unknown environment, which is modeled by a discrete-time stationary MDP $(\mathcal{S}, \mathcal{A}, p, c)$
in the present paper. Here $\mathcal{S}$ stands for the state space and $\mathcal{A} $ denotes the action space, both of which can be either discrete or continuous. A state transition probability kernel is denoted by $p,$ which is unknown to the agent. That is, for fixed $(s, a) \in \mathcal{S} \times \mathcal{A}$, $p(s'\mid s, a)$ is 
a transition probability (density) function from the state $s\in\mathcal{S}$ to the state $s'\in\mathcal{S}$ via the action $a$. 
Moreover, a cost function is denoted by $c: \mathcal{S} \times \mathcal{A}  \times\mathcal{S}\to \mathbb{R}$, where a realized cost is given by $c(s, a,s')$.  A single episode consists
of $T$ periods. Let $\mathcal{T}:=\{0, \ldots, T-1\}$ be the sequence of time spots. Conventionally, an admissible randomized (feedback) policy at $(t,s)$, denoted by $\pi(a|t,s)$, is interpreted as a probability density function (p.d.f.) on the action space $\mathcal A$ satisfying: 
\begin{itemize}
    \item $\pi(a | t, s)$ is a p.d.f. on $\mathcal{A}$ with respect to (w.r.t.) measure $\nu_{\mathcal{A}}$ for any fixed $(t,s)\in \mathcal{T}\times\mathcal{S};$
    \item $\pi(a | t, s)$ is a measurable function w.r.t. $(t,s)$ for any fixed $a\in\mathcal{A}.$
\end{itemize}
Define the set of all admissible (feedback) policies as $\Pi$. 
When the action space $\mathcal{A}$ is discrete, $\nu_{\mathcal{A}}$ is a counting measure and $\pi(a| t, s)$ means the probability mass of choosing action $a$ in state $s$ at time $t.$ When the action space $\mathcal{A}$ is continuous,  $\nu_{\mathcal{A}}$ is the Lebesgue measure. 

\subsection{Risk Measures and Scoring Functions}
\label{sec:riskmeasure}
Let $(\Omega,\mathscr{F},\mathbb{P})$ be a probability space and define $\mathcal{Z}:=\mathcal{L}_m(\Omega,\mathscr{F},\mathbb{P})$ with $m\in [1,\infty],$ which is the space of measurable random variables with finite $m$-th moment (essentially bounded measurable random variables if $m = \infty$). A risk measure  $\rho:\mathcal{X}\to\mathbb{R}$ is a functional mapping from a random variable to a real number representing its risk value, where $\mathcal{X}$ is a space of random variables. From now on, we assume that $\mathcal{X}=\mathcal{Z}$ (one may fix an appropriate $m$ for integrability according to the risk objective) and the random variable $Y\in\mathcal{X}$ is interpreted as a random cost. For example, with $m=1$, for a level $\alpha \in (0,1]$, the risk measures of VaR and ES are defined as 
$$
\VaR _{\alpha} (Y) \triangleq \inf \{x \in \mathbb R :   \p(Y\le  x) \geq \alpha  \}, \;\;\;\; 
\ES_{\alpha}(Y) \triangleq  \frac{1}{1-\alpha}  \int_{\alpha}^{1} \VaR_\beta (Y) \d \beta,  \quad Y\in \mathcal X.\notag
$$
In the context of quantitative risk management, a scoring function is a metric to evaluate and compare competing risk forecasts. As a definition (Chapter 9 of \cite{MFE15}), a classical scoring function $f:\mathbb{R}\times\mathbb{R}\to\mathbb{R}$ satisfies: 
\begin{itemize}
\item for any $y$ and $v$, $f(y,\upsilon)\geq0$ and $f(y,\upsilon)=0$ if and only if $y=\upsilon;$
\item for a fixed $y$, $f(y,\upsilon)$ is increasing in $\upsilon$ for $\upsilon>y$ and decreasing for $\upsilon<y;$
\item for a fixed $y$, $f(y,\upsilon)$ is continuous in $\upsilon.$
\end{itemize}
An elicitable risk measure is defined by 
$
\arg\min_{\upsilon
 \in\mathbb{R}}\E [f(Y,\upsilon)]$,  
where $f$ is a classical scoring function (see \cite{gneiting2011making,FLWW2024}). Expectation, VaR, and expectile are common examples of elicitable risk measures (not examples of our risk objective in later Problems \eqref{prob:SRP} and \eqref{problem_to_solve}).

In the present paper, we consider convex scoring functions, which are defined in Definition \ref{def1}. A convex scoring function is bivariate and convex w.r.t. its second argument. The two arguments of a convex scoring function are interpreted as the cost value $y$ and the auxiliary variable $\upsilon,$ respectively. 

\begin{definition}[Convex scoring function]
\label{def1}
A function $f: \mathcal{D} \times \R \rightarrow \mathbb{R}$ is called a convex scoring function if $\mathcal{D}\subset \mathbb{R}$ is the domain of $y$ and $f(y,\upsilon)$ is convex in $\upsilon$ for any fixed $y\in\mathcal{D}$. 
\end{definition}
Definition \ref{def1} is quite general because the convexity is closely tied to unimodality. If the convexity is replaced by the condition that $-f(y,\upsilon)$ is unimodal in $\upsilon$ (i.e., for a fixed $y$, there exists some $\upsilon_0\in\mathbb{R}$ such that $f(y,\cdot)$ is decreasing on $(-\infty,\upsilon_0]$ and increasing on $[\upsilon_0,\infty)$), then the scoring function satisfies the second condition in the definition of the classical scoring function above; see also \cite{bellini2015elicitable}. In our paper, we focus on the convexity-based formulation to enhance the clarity and facilitate some mathematical arguments. The next result for stochastic optimization readily follows from Definition \ref{def1}.

\begin{lemma}
\label{simpleProposition1}
For any random variable $Y\in\mathcal{X},$ $\E[ f\left(Y,\upsilon\right)]$ is convex in $\upsilon.$
\end{lemma}

\section{Problem Formulation}\label{sec:model}

In the present paper, the underlying MDP is assumed to be unknown, and we are interested in the RL problem by adopting the static risk measure on the accumulated costs and aim to optimize 
\begin{equation}\label{prob:SRP}
    \inf_{\pi\in\Pi}\rho\left(\sum_{t=0}^{T-1}c(S_t,A_t,S_{t+1})\right),
\end{equation}
which is a risk-sensitive RL problem under a static risk problem (SRP) formulation. \cite{tamar2016sequential} provides a policy gradient approach for SRP when $\rho$ is a coherent risk measure with a risk envelope of a general form (see Assumption II.2 in \cite{tamar2016sequential}). Here we consider $\rho$ in another class of risk measures associated to convex scoring functions: 
\begin{equation}\label{eq:scoring}
    \rho(Y) = \inf_{\upsilon \in \mathbb{R}}h(\E[ f\left(Y,\upsilon\right)],\upsilon),
\end{equation}
where the function $f$ is a convex scoring function and $h$ is a bivariate transform function strictly increasing w.r.t. its first argument, and the second argument $\upsilon$ is the auxiliary variable. Such risk measures are not required to be coherent or even monetary, but they encompass a broad class of functionals and many cases of coherent risk measures, such as expectation $\rho(Y)=\inf_{\upsilon\in\mathbb{R}}\E [Y]$ and Expected Shortfall (ES) $\rho(Y) = \inf_{\upsilon\in\mathbb{R}} \E[\upsilon+\frac{1}{\alpha}(Y-\upsilon)^+].$ 
Many convex (but may not be coherent) risk measures are also covered, such as 
entropic risk measures
$\rho(Y) = \inf_{\upsilon\in\mathbb{R}} \frac{1}{\gamma} \log \left(\E [\exp(\gamma Y)]\right)$ with a risk-aversion parameter $\gamma > 0$ 
and expected utility $\rho(Y) = \inf_{\upsilon\in\mathbb{R}} \E [u(Y)]$ with a general convex function $u$ (see \cite{rockafellar2013fundamental}). 

Note that some risk measures in our consideration are not even monetary, such as median absolute deviation (MAD) $\rho(Y) =\inf_{\upsilon\in\mathbb{R}} \E [|Y-\upsilon|],$ variance $\rho(Y)  =\inf_{\upsilon\in\mathbb{R}} \E [|Y-\upsilon|^2]$ and asymmetric variance $\rho(Y) = \inf_{\upsilon \in \mathbb{R}} \, \mathbb{E} \left[ 
\alpha (Y - \upsilon)^2 \cdot \mathds{1}_{\left\{Y \geq \upsilon\right\}} + 
(1 - \alpha)(Y - \upsilon)^2 \cdot \mathds{1}_{\left\{Y \leq \upsilon\right\}} 
\right]$. Another important class of risk measures is the family of mean-risk utilities in Example \ref{rm_exam2}, which are constructed as the combination of an expectation and a risk penalty. 
\begin{example} \label{rm_exam2}
The mean-risk utility is defined as
\begin{align}
        \rho(Y) &= \E [Y] + \lambda\tilde{\rho}(Y),\notag
\end{align}
where $\lambda>0$ and $\tilde{\rho}$ can be chosen as Expected Shortfall, MAD, variance, and asymmetric variance.
\end{example}
The mean-variance utility in Example \ref{rm_exam2} is of great importance in portfolio optimization, asset allocation, and operations research. \cite{zhang2021mean} and \cite{huang2022achieving} provide RL methods for mean-variance objectives in discrete and continuous time, respectively. Our approach differs from both and can accommodate various risk penalty terms. 

As another important class of risk measures, the OCE (see \cite{ben2007old}) is defined as $\rho(Y)=\min_{\upsilon\in\mathbb{R}}\left\{\upsilon+\E[u(Y-\upsilon)]\right\}$ for some concave function $u(\cdot).$ By taking $f(y,\upsilon)=u(y-\upsilon)$ and $h(x,\upsilon)=\upsilon+x,$ these OCE risk measures can be represented using our formulation in Eq. \eqref{eq:scoring}. Examples beyond the OCE class include the EVaR (Example \ref{rm_exam3}) and some utility-based risk measures  (e.g., $\rho(Y)=\inf_{\upsilon\in\mathbb{R}}\E \left[{Y^2}e^{-\upsilon}+e^{\upsilon}\right]$); see Chapter 4 of \cite{FS16}. 
\begin{example}
\label{rm_exam3}
The entropic Value-at-Risk (EVaR) is defined as 
\begin{align}
        \rho(Y) = \inf_{t>0}\frac{1}{t}\log\left(\frac{1}{\alpha} \E [e^{tY}] \right)= \inf_{\upsilon\in\mathbb{R}}\frac{1}{e^\upsilon}\log \left(\frac{1}{\alpha}\E [e^{e^\upsilon Y}]\right).\notag
\end{align}
The scoring function is $f(y,\upsilon)=\frac{1}{\alpha}\exp(e^{\upsilon} y)$, which is convex if $y \in [0,\infty)$ (i.e., $Y$ is nonnegative  almost surely).
\end{example}

Note that the RL problem in \eqref{prob:SRP} is time-inconsistent, i.e., the learned optimal policy at time $t$ may no longer hold its optimality at a later time. To resolve the time-inconsistent issue, we first rewrite the dynamic optimization problem as a two-stage optimization problem
\begin{equation}
\label{problem_to_solve}
\begin{aligned}
\inf_{\pi \in \Pi} \rho \left(  \sum_{t = 0}^{T-1}c(S_t, A_t,S_{t+1}) \right)
=&\inf_{\pi \in \Pi} \inf_{\upsilon \in \R} h\left(\E \left[f\left(\sum_{t = 0}^{T-1}c(S_t, A_t,S_{t+1}),\upsilon\right)\right] ,\upsilon\right) \\
=& \inf_{\upsilon \in \R}\inf_{\pi \in \Pi} h\left(\E \left[f\left(\sum_{t = 0}^{T-1}c(S_t, A_t,S_{t+1}),\upsilon\right)\right],\upsilon\right)\\
=& \inf_{\upsilon \in \R} h\left(\inf_{\pi \in \Pi}\E \left[f\left(\sum_{t = 0}^{T-1}c(S_t, A_t,S_{t+1}),\upsilon\right)\right],\upsilon\right).
\end{aligned}
\end{equation}

We then introduce an augmented state process $\{Y_t\}_{t \in \mathcal{T}}$ for the inner dynamic optimization problem, which is defined as the negative of accumulated costs over time, i.e., 
\begin{equation}\label{eq:def_Y}
Y_t =-\sum_{\tau=0}^{t-1}\text{cost}_{\tau}, 
\end{equation}
where $\text{cost}_{\tau}$ represents the realized cost at time $\tau$.  By enlarging the state space for the decision making with the augmented state process $\left\{Y_t\right\}_{t\in\mathcal{T}}$, the inner dynamic optimization problem becomes time-consistent. Given $Y_t = y$, we can consider the policy $\pi(\cdot | t, s, y)$ induced from three states $t$, $s$ and $y$. The modified feedback policy is then defined if the following hold:
\begin{itemize}
    \item $\pi(a | t, s,y)$ is a p.d.f. on $\mathcal{A}$ w.r.t. measure $\nu_{\mathcal{A}}$ for any fixed $(t,s,y)\in \mathcal{T}\times\mathcal{S}\times\mathbb{R};$
    \item $\pi(a | t, s,y)$ is a measurable function w.r.t. $(t,s,y)$ for any fixed $a\in\mathcal{A}.$
\end{itemize}
We denote $\tilde{\Pi}$ the set of all modified feedback policies satisfying the above conditions.

We emphasize that at time $t$, the augmented state $Y_t$ does not directly affect the next state $S_{t+1}$ or the cost $C_{t+1}.$ Problem \eqref{problem_to_solve} is equivalent to 
\begin{equation}
\label{simplified_problem}
\inf_{\upsilon \in \R} h\left(\inf_{\pi \in \tilde{\Pi}}\E \left[f\left(\sum_{t = 0}^{T-1}c(S_t, A_t,S_{t+1}),\upsilon\right)\right],\upsilon\right), 
\end{equation}
which plays a key role in our RL algorithm design. At the end of this subsection, we illustrate in Figure \ref{Framework} the proposed learning procedure with the augmented state.  

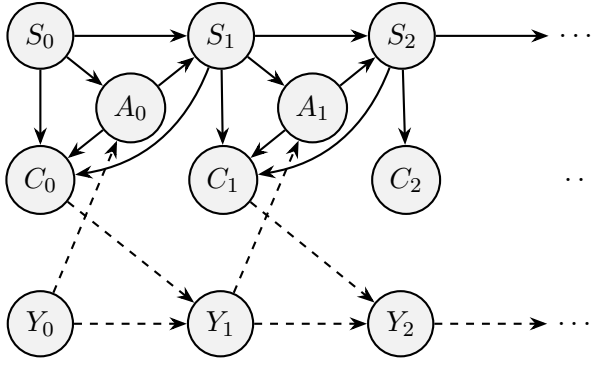
\begin{figure}[h]
\centering
\begin{tikzpicture}[
    node distance=1.0cm and 1.5cm, 
    every state/.style={thick, fill=gray!10, minimum size=0.25cm, draw},
    >=Stealth, 
    auto,
    align=center
]
    \node[state] (s0) {$S_0$};
    \node[state, right=of s0] (s1) {$S_1$};
    \node[state, right=of s1] (s2) {$S_2$};
    \node[draw=none, right=of s2] (ellipsis1) {$\cdots$};
    
    \node[state, below=of s0] (c0) {$C_0$};
    \node[state, right=of c0] (c1) {$C_1$};
    \node[state, right=of c1] (c2) {$C_2$};
    \node[draw=none, right=of c2] (ellipsis2) {$\cdots$};
    
    \node[state, below=of c0] (y0) {$Y_0$};
    \node[state, right=of y0] (y1) {$Y_1$};
    \node[state, right=of y1] (y2) {$Y_2$};
    \node[draw=none, right=of y2] (ellipsis3) {$\cdots$};
    
    \node[state] (a0) at ($(c0)!0.5!(s1)$) {$A_0$};
    \node[state] (a1) at ($(c1)!0.5!(s2)$) {$A_1$};
    
    \path[->, thick] 
        (s0) edge node[above] {}(s1)
        (s1) edge node[above] {} (s2)
        (s2) edge node[above] {} (ellipsis1)
        
        (s0) edge node[right] {} (c0)
        
        (c0) edge[->, dashed] node[right] {} (y1)
        (s1) edge[bend left=30] node[right] {} (c0)
        (c1) edge[->, dashed] node[right] {} (y2)
        (s2) edge[bend left=30] node[right] {} (c1)
        
        (s1) edge node[above] {} (c1)
        (s2) edge node[above] {} (c2)
        
        (y1) edge[->, dashed] node[above] {} (y2)
        (y2) edge[->, dashed] node[above] {} (ellipsis3)
        (y0) edge[->, dashed] node[above] {} (y1)
        
        (a0) edge[->] (c0)
        (a0) edge[->] (s1)
        
        (a1) edge[->] (c1)
        (a1) edge[->] (s2)

        (y0) edge[->, dashed] (a0)
        (y1) edge[->, dashed] (a1)
        
        (s0) edge[->] (a0)
        (s1) edge[->] (a1);

\end{tikzpicture}
\caption{Proposed Learning Procedure}
\label{Framework}
\end{figure}

\section{Methodology}\label{sec:method}

We assume that $\mathcal{S}\times\mathcal{A}$ is a compact domain and $c$ is bounded on $\mathcal{S}\times\mathcal{A}\times\mathcal{S}.$ Then the range of $y$  is restricted to the compact set $\mathcal{Y} \triangleq [\underline{y}, \bar{y}]$ with
$$
\underline{y} < T\cdot\inf_{s,s' \in \mathcal{S}, a \in \mathcal{A}} c(s, a,s')\in\mathbb{R}, \;\;\;\; \bar{y} > T\cdot\sup_{s,s' \in \mathcal{S}, a \in \mathcal{A}} c(s, a,s')\in\mathbb{R}.
$$
However, we do not assume any continuity or semi-continuity on the transition probability $p$ or the cost function $c.$ This differs significantly from the vast majority of the literature, in which the theoretical result is often based on the continuous MDP assumption (see Assumption 3.1 in \cite{bauerle2022markov} and Assumption C in \cite{bauerle2011markov}). \cite{CJ2024} and \cite{tamar2016sequential} also implicitly use the continuity of transition probabilities. The case with discontinuities in risk-sensitive reinforcement learning is underexplored. To ensure the feasibility of the approximation, we introduce an additional Assumption \ref{ass1} on the risk measure $\rho$, specifically concerning the scoring functions $f$ and $h$, with some mild requirements.

\begin{assumption}[Continuity of $f$ and $h$]
\label{ass1}
There exists a compact set $\mathcal{D}\subset\mathbb{R}$ containing $\mathcal{Y}$ such that
\begin{itemize}
\item[(1)] $f(y,\upsilon)$ is continuous on $\mathcal{D}\times\mathbb{R}$; and 
\item[(2)] $h$ is continuous on $f(\mathcal{D},\mathbb{R})\times\mathbb{R},$ where $f(\mathcal{D},\mathbb{R})\triangleq \left\{f(y,\upsilon):y\in\mathcal{D},\upsilon\in\mathbb{R}\right\}.$
\end{itemize}
\end{assumption}

\begin{remark}
If Assumption C in \cite{bauerle2011markov} is further satisfied, i.e., $p$ is continuous in $(s,a)$ and $c$ is lower semi-continuous, we can remove Assumption \ref{ass1}. Furthermore, under Assumption C in \cite{bauerle2011markov}, our approach also accommodates risk measures induced from classical scoring functions.
\end{remark}


We need to further restrict the range of $\upsilon$ to a compact set such that it contains at least one optimal solution $\upsilon^*$ of Problem \eqref{simplified_problem}. To this end, we first show in Proposition \ref{simpleProposition2} below that the range of $\upsilon$ can be restricted to a compact set. 
However, in some cases, the conditions in Proposition \ref{simpleProposition2} may not be satisfied (see Remark \ref{entropic_var_case}) and we have to assume that the minimum of the objective is attained on $[\underline{\upsilon},\bar{\upsilon}]$ for some pre-known $\underline{\upsilon},\bar{\upsilon}\in\mathbb{R}.$ This assumption is commonly used in some existing studies in RL; see \cite{baxter2001infinite}. 
Note that even when this assumption is not fulfilled, we can still find $\epsilon$-correct range for $\upsilon,$ i.e., that for any $\epsilon>0,$ there exists an interval $[\underline{\upsilon}^{\epsilon},\bar{\upsilon}^{\epsilon}]$ such that
\begin{equation}
    \min_{\upsilon \in [\underline{\upsilon}^{\epsilon},\bar{\upsilon}^{\epsilon}]} h(\E[ f\left(Y,\upsilon\right)],\upsilon)\leq \inf_{\upsilon\in\mathbb{R}}h(\E[ f\left(Y,\upsilon\right)],\upsilon) +\epsilon,\notag
\end{equation}
for any random variable $Y$ satisfying $\underline{y}\leq Y \leq \bar{Y}$ almost surely.
To address a similar issue, \cite{ni2022policy} derives a lower bound of the optimal $\upsilon$ (which agrees with Remark \ref{entropic_var_case}) and imposes an assumption on its upper bound for the entropic VaR. 

\begin{proposition}
     \label{simpleProposition2}
Assume that there exists a strictly increasing function $g$ such that
\begin{itemize}
\item[(1)] $g(h(\cdot,\cdot))$ is convex on $f(\mathcal{D})\times\mathbb{R}$; and  
\item[(2)]  for any $y\in \mathcal[\underline{y},\bar{y}],$ the minimal value of $h(f(y,\cdot),\cdot)$ can be attained at some $\upsilon^*\in\mathbb{R}$; and
\item[(3)] it holds that \begin{equation}
\label{condition3_lim}
\lim_{\upsilon\to \pm\infty}g\left(h\left(\min_{y\in[\underline{y},\bar{y}]}f(y,\upsilon),\upsilon\right)\right) = \infty,
\end{equation}
or $h(y,\cdot)$ is constant for any fixed $y\in[\underline{y},\bar{y}],$ and $f(y,\upsilon) =w(y)+\phi(\upsilon-\psi(y)),$ for some continuous functions $w,\phi,\psi.$
\end{itemize}
Then there exist some $m^{\underline{y},\bar{y}},M^{\underline{y},\bar{y}}\in\mathbb{R}$ such that $m^{\underline{y},\bar{y}}<M^{\underline{y},\bar{y}}$ and the minimum of $h(\E[ f\left(Y,\cdot\right)],\cdot)$ is attained on $
\left[m^{\underline{y},\bar{y}},M^{\underline{y},\bar{y}}\right],$  for any random variable $Y$ satisfying  $\underline{y}\leq Y\leq \bar{y}$ almost surely. 
\end{proposition}
\begin{remark}
If $f$ and $h$ are constants of $\upsilon,$ i.e., there admit univariate functions $\tilde{h}$ and $\tilde{f}$ such that $\rho(Y) = \tilde{h}(\E[\tilde{f}(Y)]),$ then we take $g = \tilde{h}^{-1}$ and the conditions (1) and (2) are satisfied (since $h$ is strictly increasing w.r.t. its first argument, $\tilde{h}$ must be strictly increasing). We further take $w = \tilde{f},\phi=\psi=0$ and the condition (3) is satisfied. This is true for the expectation and the entropic risk measure.

If $h(x,\upsilon) = \tilde{h}(x),$ i.e. $\rho(Y)=\tilde{h}(\min_{\upsilon\in\mathbb{R}}\E [f(Y,\upsilon)]),$ then we take $g(x) = \tilde{h}^{-1}(x)$ and Assumption \ref{ass1} is sufficient for condition (1). In this case, condition (2) is equivalent to the statement that $f(y,\upsilon)$ has a minimum for any fixed $y\in[\underline{y},\bar{y}],$ which is true for Expected Shortfall, variance, MAD, and mean-risk utility. \eqref{condition3_lim} is true for Expected Shortfall, MAD, variance, asymmetric variance, and mean-variance utility. 
\end{remark}
\begin{remark}
\label{entropic_var_case}
Consider the entropic VaR mentioned in Example \ref{rm_exam3}. With $f(y,\upsilon)=\frac{1}{\alpha}\exp(e^{\upsilon}y)$, $h(x,\upsilon)=e^{-\upsilon}\log(x)$ and $g(x)=\exp(x),$ condition (1) holds for the entropic VaR if $[\underline{y},\bar{y}]\subset [0,\infty)$. However, conditions (2) and (3) are not satisfied because $g(h(f(y,\upsilon),\upsilon))$ is strictly decreasing for $\upsilon\in\mathbb{R}$ for fixed $y.$ In this case, for any $\epsilon>0,$ we can find $[\underline{\upsilon}^{\epsilon},\bar{\upsilon}^{\epsilon}]$ such that
\begin{equation}
    \min_{\upsilon \in [\underline{\upsilon}^{\epsilon},\bar{\upsilon}^{\epsilon}]} h(\E[ f\left(Y,\upsilon\right)],\upsilon)\leq \inf_{\upsilon\in\mathbb{R}}h(\E[ f\left(Y,\upsilon\right)],\upsilon) +\epsilon,\notag
\end{equation}
for any $Y$ satisfying $\underline{y}\leq Y \leq \bar{Y}$ almost surely.
\end{remark}

For a fixed $\upsilon\in \Upsilon,$ denote by $\mathbb{E}_{\pi} $ the expectation under policy $\pi$ and by $\mathbb{P}_{\pi} $ the probability measure under policy $\pi.$ The value function under policy $\pi$ is defined by 
\begin{equation}
\label{hh}
    V^{\pi}_{t,\upsilon}(s,y) = \mathbb{E}_{\pi} \left.\left[f\left(\sum_{k=t}^{T-1}c(S_k,A_k,S_{k+1})-Y_t,\upsilon\right)\right|S_t=s,Y_t=y\right],\  \text{for} \ t=T-1,\ldots,0
\end{equation}
and 
\begin{equation}
    V^{\pi}_{T,\upsilon}(s,y) = f(-y,\upsilon).\notag
    \end{equation}
Then we have the recursive formulation of the value function
\begin{align}
    V^{\pi}_{t,\upsilon}(s,y) &= \mathbb{E}_{\pi} \left.\left[f\left(\sum_{k=t+1}^{T-1}c(S_k,A_k,S_{k+1})-\left(y-c\left(S_t,A_t,S_{t+1}\right)\right),\upsilon\right)\right|S_t = s,Y_t = y\right]\notag\\
    & = \mathbb{E}_{\pi} \left.\left[V^{\pi}_{t+1,\upsilon}\left(S_{t+1},y-c\left(s,A_t,S_{t+1}\right)\right)\right|S_t=s,Y_t=y\right]\notag\\
   & = \mathbb{E}_{\pi} \left.\left[V^{\pi}_{t+1,\upsilon}\left(S_{t+1},Y_{t+1}\right)\right|S_t=s,Y_t=y\right].\label{resur} 
\end{align}
The idea behind this formula is similar to the approach proposed by \cite{bauerle2011markov} to reduce the MDP under ES. However, their approach relies on the specific form of $(y-\upsilon)^+$ in the scoring function for ES, which facilitates the simple fusion of $Y$ and $\upsilon$ into one state process. For the convex scoring function, their approach no longer holds. However, our proposed method separates the variable $\upsilon$ and the augmented state process $Y$ such that we can devise the RL algorithm to accommodate more types of risk measures.

For a function $V:\mathcal{S}\times\mathcal{Y}\to\R,$ define the Bellman recursive operators 
\begin{align}
    \mathcal{J}_{t,\pi} V(s,y) &=\int_{\mathcal{A}} \int_{\mathcal{S}} \pi(a|t,s,y) p(s'|s,a)V\left(s',y-c\left(s,a,s'\right)\right) \nu_{\mathcal{S}}(\d s') \nu_{\mathcal{A}}(\d a), \notag\\
    \mathcal{J}_{t} V(s,y) &=  \inf_{a\in\mathcal{A}}\int_{\mathcal{S}}  p(s'|s,a)V\left(s',y-c\left(s,a,s'\right)\right) \nu_{\mathcal{S}}(\d s') .\notag
\end{align}

\begin{theorem}[Bellman Equation]
\label{t_Bellman}
For any fixed  $\upsilon \in \Upsilon,$ there exists a family of value functions $\left\{V_{t,\upsilon}^*(\cdot,\cdot)\right\}_{t=0,1\ldots,T}$ satisfying the following recursive system
\begin{equation}\label{BellmanEq}
\begin{aligned}
V_{T,\upsilon}^*(s,y) &= f(-y,\upsilon),\\
     V_{t,\upsilon}^*(s,y) &= \mathcal{J}_t V^*_{t+1,\upsilon}(s,y),\quad t=T-1,\ldots,0,
\end{aligned}
\end{equation}
and the inequality
\begin{equation}
    V_{t,\upsilon}^*(s,y) \leq  V^{\pi}_{t,\upsilon}(s,y), \notag
\end{equation}
for any $t=0,1,\ldots,T,$ $s \in \mathcal{S},$  $y \in \mathcal{Y},$ and policy $\pi\in\tilde{\Pi}.$ Furthermore, there exists a policy sequence $\left\{\pi_{\upsilon,n}\right\}^{\infty}_{n=1} \subset \tilde {\Pi}$ such that 
\begin{equation}
  \left[ V_{t,\upsilon}^{\pi_{\upsilon,n}}(s,y) -V_{t,\upsilon}^*(s,y)\right]\to 0 \ (n\to\infty),\notag
\end{equation}
for any $t=0,1\ldots,T,s\in\mathcal{S},y\in\mathcal{Y}.$
\end{theorem}
\begin{remark}
\label{remark_optimalpolicy}
If Assumption C in \cite{bauerle2011markov} is satisfied, then there exists an optimal policy $\pi^*_{\upsilon}\in\tilde{\Pi}$ such that 
\begin{equation}
   V_{t,\upsilon}^{\pi^*_{\upsilon}}(s,y) =V_{t,\upsilon}^*(s,y)\notag
\end{equation}
for any $t\in\mathcal{T},s\in\mathcal{S},y\in\mathcal{Y}.$ Such a conclusion is widely discussed in the MDP-related literature (see 
\cite{sutton2018reinforcement} and \cite{bauerle2022markov}).
\end{remark}

 Combining Theorem \ref{t_Bellman}, the dominated convergence theorem and the continuity of $h,$ the next proposition holds.
\begin{proposition}
\label{prop_appro}
For any $\upsilon\in\Upsilon$ and $\epsilon>0,$ there exists a policy $\pi_{\upsilon,\epsilon} \in \tilde {\Pi}$ such that \begin{equation}
        \mathbb E [V^{\pi_{\upsilon,\epsilon}}_{0,\upsilon}(S_0,0)]-\mathbb E [V^{*}_{0,\upsilon}(S_0,0)]\leq \epsilon,\notag
    \end{equation}
    and \begin{equation}
        h(\mathbb E [V^{\pi_{\upsilon,\epsilon}}_{0,\upsilon}(S_0,0)],\upsilon)-h(\mathbb E [V^{*}_{0,\upsilon}(S_0,0)],\upsilon)\leq \epsilon.\notag
    \end{equation}
\end{proposition}

Note that \eqref{simplified_problem} is equivalent to
\begin{equation}
    \min_{\upsilon\in\Upsilon}h(\mathbb E [V^*_{0,\upsilon}(S_0,0)],\upsilon).\notag
\end{equation}
Therefore, studying the properties of $\mathbb E [V^*_{0,\upsilon}(S_0,0)]$ is of fundamental importance. Proposition \ref{propmin} establishes the continuity of $\upsilon\mapsto \mathbb E [V^*_{0,\upsilon}(S_0,0)]$, a key result in showing the approximation of our RL algorithm in later subsections. Notably, our analysis does not require a continuity assumption on the MDP, rendering the proof of Proposition \ref{propmin} technically nonstandard. 

\begin{proposition}\label{propmin}
     We have that $\upsilon\mapsto\mathbb E [V^*_{0,\upsilon}(S_0,0)]$ is continuous on $\Upsilon$ and   
    $\upsilon\mapsto h(\mathbb E[ V^*_{0,\upsilon}(S_0,0) ],\upsilon)$ has a minimizer (not necessarily unique) on $\mathbb{R}.$ Furthermore, \begin{equation}
       \min_{\upsilon\in\Upsilon} h(\mathbb E [V^*_{0,\upsilon}(S_0,0) ],\upsilon)= \min_{\upsilon\in\mathbb{R}}  h(\mathbb E[V^*_{0,\upsilon}(S_0,0)],\upsilon).\notag
    \end{equation}
Hence, we can define the set of minimizers as 
\begin{equation}
    \mathscr{U}^* = \arg\min_{\upsilon\in \Upsilon} h(\mathbb E [V^*_{0,\upsilon}(S_0,0)] ,\upsilon).\notag
\end{equation}
\end{proposition}

\subsection{Part 1: Gradient Calculation and Neural Network (NN) Training}
We adopt the policy gradient Actor-Critic RL algorithm for the inner optimization problem; see \cite{konda2003onactor}. That is, for the inner dynamic optimization problem, 
we alternatively learn two functions, the value function $V$ and the policy $\pi.$ In our work, both the value function and the policy are characterized by fully connected and multilayered feed-forward networks, for which the parameters are denoted by $\theta\in\Theta$ and $\phi\in\Phi$, respectively. The inputs to both networks include $(t, \upsilon, s, y).$ The network's outputs are continuous with respect to  $(t, \upsilon, s, y)$; 
see also some existing studies using the policy gradient approach on SRP in \cite{tamar2016sequential},  \cite{zhang2021mean}, \cite{ni2022policy} and \cite{CJ2024}. In contrast to these studies, our algorithm does not require the simulation-upon-simulation approach or the sample-average-approximation approach, which greatly reduces the computational complexity in simulating trajectories.

We denote the outputs of the two networks by $\pi^{\theta}(\cdot|t,\upsilon,s,y)$ and $V^{\phi}_{t,\upsilon}(s,y;\theta),$ respectively. We define
\begin{equation}
\label{hh}
    V_{t,\upsilon}(s,y;\theta) = \mathbb{E}_{\pi^{\theta}_{\upsilon}} \left.\left[f\left(\sum_{k=t}^{T-1}c(S^{\theta}_k,A^{\theta}_k,S^{\theta}_{k+1})-Y_t^{\theta},\upsilon\right)\right|S_t^{\theta}=s,Y_t^{\theta}=y\right], t=T-1,\ldots,0,
\end{equation}
where $\pi^{\theta}_{\upsilon}\in \tilde{\Pi}$ is defined by
\begin{equation}
    \pi^{\theta}_{\upsilon}(a|t,s,y) = \pi^{\theta}(a|t,\upsilon,s,y).\notag
\end{equation}
We have the following continuity result.
\begin{lemma}
For any fixed $s\in\mathcal{S}$ and $t\in\mathcal{T}$,   
$V_{t,\upsilon}(s,y;\theta)$ is continuous in $\upsilon\in\Upsilon$ and $y\in\mathcal{Y}$. 
\end{lemma}

In the Critic-step, we use Equation \eqref{resur} to update the value function. The following Theorem \ref{t1} shows that under some assumptions, the value function $V^{\phi}_{t,\upsilon}(s,y;\theta)$ produced from the NN can approximate the true value function $V_{t,\upsilon}(s,y;\theta)$ uniformly on $(s, y,\upsilon) \in S \times \mathcal{Y}\times\Upsilon$. The algorithm for the Critic-step is provided in Algorithm \ref{alg:alg2}.

\begin{theorem}[Approximation of the value function]
    \label{t1}
    Let $\pi^{\theta}$ be a fixed parameterized policy, with the corresponding value function defined in Equation \eqref{hh}. For any $\epsilon^*>0$ and $\gamma^*>0,$ there exists ANN parameters $\phi$ such that
    \begin{equation}
  \nu_{\mathcal{S}}\left( \left\{ s:\sup_{t\in\mathcal{T},\upsilon\in\Upsilon,y\in\mathcal{Y}} | V_{t,\upsilon}(s,y;\theta) -V^{\phi}_{t,\upsilon}(s,y;\theta)| > \gamma^*\right\}\right)< \epsilon^*.\notag
    \end{equation}
\end{theorem}
\begin{remark}
Note that we do not assume any continuity or semi-continuity on the transition probability $p$ or the cost function $c.$ If we further assume that $p$ and $c$ are both continuous in $(s,a,s'),$ the conclusion of Theorem \ref{t1} can be strengthened: there exists an ANN $V^{\phi},$ such that  
\begin{equation} \sup_{s\in\mathcal{S},t\in\mathcal{T},\upsilon\in\Upsilon,y\in\mathcal{Y}} | V_{t,\upsilon}(s,y;\theta) -V^{\phi}_{t,\upsilon}(s,y;\theta)| \leq \gamma^*.\notag
\end{equation}
\end{remark}

\begin{algorithm}[htb]
	\caption{Policy Evaluation (Critic-step)}
	\label{alg:alg2}
	\begin{algorithmic}
        \STATE \textbf{Input:} ANNs for the value function $V^{\phi}$ and the policy $\pi^{\theta},$ the replay buffer $\mathcal{M},$ epochs $K,$ the batch size $B.$
        
        \FOR {each epoch $\kappa=1$ to $K$} 
        \STATE Set the gradients of $V^{\phi}$ to zero;
        \FOR {each  $t=T-1$ to $0$} 
        \STATE Sample $B$ states $\left(s^{(t,b)},y^{(t,b)},a^{(t,b)},s^{(t+1,b)},y^{(t+1,b)},\upsilon^{(b)}\right)(b=1,\ldots,B)$ from $\mathcal{M};$
        \ENDFOR
               \FOR{each state $b=1,\ldots,B$ and  $t=T-1$ to $0$}
               \STATE Compute the predicted values $\hat{v}_t^b=V_{t,\upsilon^{(b)}}^{\phi}\left(s^{(t,b)},y^{(t,b)};\theta\right);$
               \STATE Compute the cost $c^{(t,b)}=c(s^{(t,b)},a^{(t,b)},s^{(t+1,b)});$
               \IF {$t=T-1$}
               \STATE Set the target value as
               $$v^b_t= f\left(c^{(t,b)}-y^{(t,b)},\upsilon^{(b)}\right);
               $$
               \ELSE 
               \STATE Set the target value as
               $$v^b_t= V_{t+1,y^{(b)}_0}^{\phi}\left(s^{(t+1,b)},y^{(t,b)}-c^{(t,b)};\theta\right);
               $$
               \ENDIF
               \ENDFOR
               \STATE Compute the expected square loss  $\frac{1}{BT}\sum_{b=1}^B\sum_{t=0}^{T-1}\left(v_{t}^b -\hat{v}_t^b\right)^2;$
               \STATE Update $\phi$ by performing an Adam optimizer step;
            
             \ENDFOR
            \STATE \textbf{Output:} The value function $V_{t,\upsilon}^{\phi}(s,y;\theta)$ under the policy evaluation. 
	\end{algorithmic}
\end{algorithm}

In the Actor-step, we have the following computations of the policy gradient. For $t = T-1,$ 
\begin{align}
    \nabla_{\theta}V_{T-1,\upsilon}\left(s,y;\theta\right)  =&\nabla_{\theta}\mathbb{E}_{\pi^{\theta}}\left[f\left(c\left(s,A^{\theta}_{T-1},S^{\theta}_{T}\right)-y,\upsilon\right)\right]\notag\\
    =& \nabla_{\theta}  \int_{\mathcal{S}} \int_{\mathcal{A}} \pi^{\theta}(a|t,\upsilon,s,y) p(s'|s,a)f(c(s,a,s')-y,\upsilon) \nu_{\mathcal{A}}(\d a) \nu_{\mathcal{S}}(\d s')\notag\\
    = &   \int_{\mathcal{S}} \int_{\mathcal{A}} \nabla_{\theta}\pi^{\theta}(a|t,\upsilon,s,y) p(s'|s,a)f(c(s,a,s')-y,\upsilon) \nu_{\mathcal{A}}(\d a) \nu_{\mathcal{S}}(\d s')\notag\\
    =&  \int_{\mathcal{S}} \int_{\mathcal{A}}\pi^{\theta}(a|t,\upsilon,s,y)  p(s'|s,a)f(c(s,a,s')-y,\upsilon) \nabla_{\theta}\log\pi^{\theta}(a|t,\upsilon,s,y)\nu_{\mathcal{A}}(\d a) \nu_{\mathcal{S}}(\d s')\notag\\
     =& \mathbb{E}_{\pi^{\theta}}\left.\left[\nabla_{\theta}\log \pi^{\theta}\left(A_{T-1}^{\theta}\right.|t,\upsilon,s,y\right)f\left(c\left(s,A_{T-1}^{\theta},S_{T}^{\theta}\right)-y,\upsilon\right)\right].\notag
\end{align}
For $t < T-1,$ 
\begin{align}
\label{Tpg}
    \nabla_{\theta}V_{t,\upsilon}(s,y;\theta) &= \nabla_{\theta} \int_{\mathcal{S}} \int_{\mathcal{A}} \pi^{\theta}(a|t,\upsilon,s,y)p(s'|s,a)V_{t+1,\upsilon}(s',y-c(s,a,s');\theta) \nu_{\mathcal{A}}(\d a) \nu_{\mathcal{S}}(\d s')\notag\\
    &= \int_{\mathcal{S}} \int_{\mathcal{A}} V_{t+1,\upsilon}(s',y-c(s,a,s');\theta)\nabla_{\theta}\pi^{\theta}(a|t,\upsilon,s,y) p(s'|s,a)  \nu_{\mathcal{A}}(\d a) \nu_{\mathcal{S}}(\d s')\notag\\
    &\qquad+\int_{\mathcal{S}} \int_{\mathcal{A}}\pi^{\theta}(a|t,\upsilon,s,y)p(s'|s,a)\nabla_{\theta}V_{t+1,\upsilon}(s',y-c(s,a,s');\theta)  \nu_{\mathcal{A}}(\d a) \nu_{\mathcal{S}}(\d s')\notag\\
     &=  \mathbb{E}_{\pi^{\theta}}\left.\left[\nabla_{\theta}\log \pi^{\theta}\left(A_{t}^{\theta}\right |t,\upsilon,s,y\right)
     V_{t+1,\upsilon}\left(S^{\theta}_{t+1},y-c\left(s,A_t^{\theta},S_{t+1}^{\theta}\right);\theta\right)\right]\notag\\
&\qquad+\mathbb{E}_{\pi^{\theta}}\left[\nabla_{\theta}V_{t+1,\upsilon}\left(S^{\theta}_{t+1},y-c\left(s,A_t^{\theta},S_{t+1}^{\theta}\right);\theta\right)\right].
\end{align}

When optimizing the policy $\pi^{\theta}$, the value function $V^{\phi}$ is treated as a fixed constant, implying that its parameter $\phi$ does not vary during the computation of the policy gradient. This constancy ensures that the gradient of $V^{\phi}$ with respect to $\phi$ does not contribute to the policy optimization process. As a result, the second term in Equation \eqref{Tpg} vanishes. This principle is commonly adopted in Actor-Critic methods. For a detailed discussion of this approach, see \cite{sutton2018reinforcement} and \cite{peters2008natural}. The algorithm for the Actor-step is provided in Algorithm \ref{alg:alg3}.

We call $\theta$ is a $\gamma$-optimal parameter if 
\begin{equation}
\sup_{\upsilon\in \Upsilon}|\mathbb E [V_{0,\upsilon}(S_0,0;\theta)]-\mathbb E[V^*_{0,\upsilon}(S_0,0)]| \leq \gamma.\notag
\end{equation}
From \cite{agarwal2021theory}, we know that as long as we sample $\upsilon$ from a distribution that has a positive p.d.f. on $\Upsilon,$ $\theta$ will converge to a $\gamma$-optimal parameter for any $\gamma>0.$

\begin{algorithm}[htb]
	\caption{Policy Gradient (Actor-step)} 
	\label{alg:alg3}
	\begin{algorithmic}
       \STATE \textbf{Input:} ANNs for the value function $V^{\phi}$ and the policy $\pi^{\theta},$ the replay buffer $\mathcal{M},$ epochs $K,$ the batch size $B.$

            \FOR {each epoch $\kappa=1$ to $K$} 
                \STATE Set the gradients of $\pi^{\theta}$ to be zero;
                \FOR {each $t=T-1$ to $0$} 
        \STATE Sample $B$ elements $\left(s^{(t,b)},y^{(t,b)},a^{(t,b)},s^{(t+1,b)},y^{(t+1,b)},\upsilon^{(b)}\right)(b=1,\ldots,B)$ from $\mathcal{M};$
        \ENDFOR
               \FOR{each state $b=1,\ldots,B$ and  $t=T-1$ to $0$}
               \STATE Obtain $\hat{z}^{(t,b)}=\nabla_{\theta}\log \pi^{\theta}\left.\left(a^{(t,b)}\right|t,\upsilon^{(b)},s^{(t,b)},y^{(t,b)}\right);$
              \IF {$t=T-1$}
               \STATE Calculate the gradient
               $$l_t^b=\hat{z}^{(t,b)}f\left(c^{(t,b)}-y^{(t,b)},\upsilon^{(b)}\right);
               $$
               \ELSE 
               \STATE Calculate the gradient
               $$l_t^b= \hat{z}^{(t,b)}V_{t+1,y^{(b)}_0}^{\phi}\left(s^{(t+1,b)},y^{(t,b)}-c^{(t,b)};\theta\right);
               $$
                \ENDIF
               \ENDFOR
               \STATE Take the average  $l = \frac{1}{BT}\sum_{b=1}^B\sum_{t=0}^{T-1}l_t^b;$
               \STATE Update $\theta$ by performing an Adam optimizer step;
            
             \ENDFOR
            \STATE \textbf{Output:} An updated policy $\pi^{\theta}.$
	\end{algorithmic}
\end{algorithm}

\subsection{Part 2: Selection of the Optimal Auxiliary Variable}\label{sec:select}

We next turn to the outer optimization problem and employ the SGD to solve
\begin{equation}
\label{ystar}
    \upsilon^*(\theta,\phi) \in \arg\min_{\upsilon\in\Upsilon} \left\{ h(\mathbb{E}[V^{\phi}_{0,\upsilon}(S_0,0;\theta)],\upsilon)\right\}.
\end{equation}
Then the optimal learned policy is given by
$
    \pi^{\theta}\left(\cdot|0,\upsilon^*(\theta,\phi),s,0\right)\notag
$ at time $0$ and
and $
    \pi^{\theta}\left(\cdot|t,\upsilon^*(\theta,\phi),s,y_t\right)  \notag
$ at time $t\geq 1 $. The algorithm to search for the optimal auxiliary variable $\upsilon^*$ is provided in Algorithm \ref{alg:alg4}.

\begin{algorithm}[htb]
	\caption{Searching the optimal auxiliary variable $\upsilon^*$} 
	\label{alg:alg4}
	\begin{algorithmic}
        \STATE \textbf{Input:} ANNs for the value function $V^{\phi},$ 
        epochs $K,$ simulations $M.$
        \STATE Initialize $\upsilon^*=0;$
        \STATE Sample $M$ states $s^{(m)}_0$ from the initial distribution;
            \FOR {each epoch $\kappa=1$ to $K$} 
            \STATE Calculate $$l=\frac{1}{M}\sum_{m=1}^{M}h\left(V^{\phi}_{0,\upsilon^*}\left(s^{(m)}_0,0;\theta\right),\upsilon^*\right);$$
                \STATE Update $\upsilon^*$ by performing an SGD optimizer step;
             \ENDFOR
            \STATE \textbf{Output:} An optimal $\upsilon^{*}.$
	\end{algorithmic}
\end{algorithm}

Next, we examine the correctness of our calculation of $\upsilon^*(\theta,\phi)$ in Equation \eqref{ystar}. 
We evaluate the accuracy of $\upsilon^*(\theta,\phi)$ by using the distance from the point $\upsilon^*(\theta,\phi)$ to the set $\mathscr{U}^*$, which is defined as 
\begin{equation}
    d\left(\mathscr{U}^*,\upsilon^*(\theta,\phi)\right) = \inf \left\{\left|\upsilon-\upsilon^*(\theta,\phi)\right|:\upsilon\in \mathscr{U}^*\right\}.\notag
\end{equation}
This type of distance is reasonable in our context, because for any $\upsilon\in\mathscr{U}^*,$ there exists an asymptotically optimal policy sequence $\left\{\pi_{\upsilon,n}\right\}_{n=1}^{\infty}$ to Problem \eqref{simplified_problem}. On the other hand, as $\pi^{\theta}$ is continuous w.r.t. $\upsilon,$ $\pi^{\theta}_{\upsilon^*(\theta,\phi)}$ approximates $\pi^{\theta}_{\upsilon}$ well for some $\upsilon\in\mathscr{U}^*$ if the distance is small. 

We now turn to the conditions under which 
$\upsilon^*(\theta,\phi)$ approximates the optimal set $\mathscr{U}^*$ sufficiently well. The following lemmas provide the foundational steps toward this goal.
\begin{lemma}
    \label{l0}
For any $\gamma'>0,$ there exists some $\gamma''>0$ such that 
\begin{equation}
\underset{\upsilon \in \Upsilon}{\text{sup}} |h(\mathbb E [V^*_{0,\upsilon}(S_0,0)],\upsilon)- h(\mathbb E [V_{0,\upsilon}^{\phi}(S_0,0;\theta)],\upsilon)|    < \gamma'\notag
\end{equation}
holds for any parameters $(\theta,\phi)$ satisfying
\begin{equation}
\underset{\upsilon \in \Upsilon}{\text{sup}} |\mathbb E [V^*_{0,\upsilon}(S_0,0)] - \mathbb E [V_{0,\upsilon}^{\phi}(S_0,0;\theta)]|    < \gamma''.\notag
\end{equation} 
\end{lemma}
Lemma \ref{l0} demonstrates that the discrepancy between the expected values of the true and approximated value functions can be controlled by carefully selecting the parameters $(\theta,\phi)$. Building on this, Lemma \ref{l2} establishes a direct connection between the accuracy of $h(\cdot, \upsilon)$ and the proximity of $\upsilon^*(\theta,\phi)$ to $\mathscr{U}^*$. 
\begin{lemma}
  \label{l2} 
For any $\gamma>0,$ there exists some $\gamma'>0$ such that
\begin{equation}
   d\left(\mathscr{U}^*,\upsilon^*({\theta},\phi)\right)\leq \gamma\notag
\end{equation}
holds for any parameters $(\theta,\phi)$ satisfying
\begin{equation}
\label{condition11}
    \underset{\upsilon \in \Upsilon}{\text{sup}}\  |h(\mathbb E[V^*_{0,\upsilon}(S_0,0) ],\upsilon)- h(\mathbb E [V_{0,\upsilon}^{\phi}(S_0,0;\theta)],\upsilon)|    < \gamma'.
\end{equation}

\end{lemma}

The results of Lemmas \ref{l0} and \ref{l2} describe how the parameters $(\theta,\phi)$ may influence the optimality of $\upsilon^*(\theta,\phi)$. Theorem \ref{t2} synthesizes these insights into a practical criterion for selecting $\upsilon^*(\theta,\phi)$. By imposing conditions on the optimality of $\theta$ and the approximation quality of $\phi$, the theorem ensures the proximity of $\upsilon^*(\theta,\phi)$ to $\mathscr{U}^*$ in the sense of the distance metric. 
\begin{theorem}[Selection of the optimal $\upsilon$]
\label{t2}
If the distribution of $S_0$ has a p.d.f. $p(\cdot)$ w.r.t. $\nu_{\mathcal{S}},$ then for any $\gamma>0,$ there exist some $\epsilon_1',\gamma'_1>0$ such that 
\begin{equation}
    d\left(\mathscr{U}^*,\upsilon^*(\theta,\phi) \right)\leq\gamma\notag
\end{equation}holds for any parameters $(\theta,\phi)$ satisfying
\begin{itemize}
    \item[(1)] $\theta$ is a $\gamma'_1$-optimal parameter;
    \item[(2)] $\phi$ satisfies that $\nu_{\mathcal{S}}\left(  \sup_{t\in\mathcal{T}, y\in\mathcal{Y},\upsilon\in\Upsilon}| V_{t,\upsilon}(s,y;\theta) -V^{\phi}_{t,\upsilon}(s,y;\theta)|>\gamma'_1 \right)<\epsilon_1'.$
\end{itemize}
\end{theorem}
The above theorem highlights the importance of the initial state distribution in determining the conditions for optimality. As elaborated in the subsequent remark, the case where $S_0$ is deterministic (i.e., a point mass) requires a more stringent condition on the approximation error of $\phi$. 
\begin{remark}
    If $S_0$ has a point mass at $s_0,$ i.e., $S_0=s_0$ almost surely, then the condition (2) should be replaced by \begin{equation}
        \sup_{t\in\mathcal{T}, y\in\mathcal{Y},\upsilon\in\Upsilon}| V_{t,\upsilon}(s,y;\theta) -V^{\phi}_{t,\upsilon}(s,y;\theta)|\leq\gamma'_1. \notag
    \end{equation}
\end{remark}

\subsection{Part 3: Sampling the auxiliary variable $\upsilon$}\label{sec:sample}
Regarding the computational cost, the method proposed in this paper avoids the simulation-upon-simulation convention but expands $Y_t$ and $\upsilon$ as inputs to the neural network as a tradeoff. Although we only use $Y_t$ as the augmented state of the MDP, learning based on the sampled $\upsilon$ also adds computational complexity in both Critic and Actor steps. Therefore, it is important to explore an efficient sampling of $\upsilon$ that makes the agent quickly converge to the optimal policy.

\cite{agarwal2021theory} shows that sampling from any distribution that has a positive p.d.f. on $\Upsilon$ will lead to a $\gamma$-optimal parameter. We can simply choose a uniform distribution on $\Upsilon.$ Here we propose an alternative approach with a higher efficiency. For any initial $\upsilon^*_0$ and parameters $\phi_0$ and $\theta_0,$ we sample $\upsilon$ and update the parameters according to the following process: in stage $n\ (n\geq 1),$ we 
\begin{equation}
    \begin{aligned}
    & \text{Sample } \upsilon \sim \mathrm{N} \left(\upsilon^*_{n-1},\sigma_n^2\right)\text{ and generate trajectories};\notag \\
    &  \text{Perform $L(n)$ Critic-steps and get $\phi_n$};\notag\\
    &  \text{Perform $L(n)$ Actor-steps and get $\theta_n$};\notag\\
    &  \text{Calculate } \upsilon^*_n\in\arg\min_{\upsilon\in\Upsilon} \left\{ h(\mathbb{E}\left[V^{\phi_n}_{0,\upsilon}(S_0,0;\theta_n)\right],\upsilon)\right\}.\notag
\end{aligned}\end{equation}
Both $L(n)$ and $\sigma^2_n$ are decreasing w.r.t. $n$ in order to fully explore in early iterations and ensure a faster convergence in later iterations. This sampling method is related to the alternating minimization algorithm. We sample $\upsilon$ from a normal distribution with a mean of $\upsilon^*_n$ and a small variance at epoch $n$. This leads to a faster convergence of $\pi^{\theta}_{\upsilon^*_n}$ to minimize $h(\E [V_{0,\upsilon^*_n}(S_0,0;\theta)],\upsilon)$ compared to sampling uniformly over $\Upsilon.$ Theorem \ref{thm:am_converge} shows that under some regularization conditions, this method will lead to a solution of \eqref{simplified_problem} even without any exploration on $\upsilon.$ The proof of Theorem \ref{thm:am_converge} is imitated from \cite{niesen2007adaptive} and \cite{csiszar1984information}.

\begin{theorem}
\label{thm:am_converge}
For any $\epsilon>0,$ $\upsilon_1\in\Upsilon,$  
\begin{equation}
   \pi_{n,\epsilon}\in \left\{\pi\in\tilde{\Pi}:h(\E [V^{\pi}_{0,\upsilon_n}(S_0,0)],\upsilon_n)-h(\E [V^{*}_{0,\upsilon_n}(S_0,0)],\upsilon_n)<\frac{\epsilon}{n^2}\right\},\notag
\end{equation}
and \begin{equation}
\label{argmin}
    \upsilon_{n+1} \in\arg\min_{\upsilon\in\Upsilon} h(\E [V^{\pi_{n,\epsilon}}_{0,\upsilon}(S_0,0)],\upsilon),\notag
\end{equation}
we have $\lim_{n\to\infty}h(\E [V^{\pi_{n,\epsilon}}_{0,\upsilon_{n+1}}(S_0,0)],\upsilon_{n+1})$ exists finitely. Furthermore, if there exists a nonnegative function $\delta:\Upsilon\times\Upsilon\to[0,\infty)$ such that for any $\upsilon'\in\Upsilon,\ \pi\in\tilde{\Pi},$ the following two conditions hold: 
\begin{itemize}
    \item[(1)] for any $\upsilon\in\arg\min_{\upsilon\in\Upsilon}h(\E [V^{\pi}_{0,\upsilon}(S_0,0)],\upsilon), $ it holds that
\begin{equation}
    \delta(\upsilon',\upsilon) + h(\E [V^{\pi}_{0,\upsilon}(S_0,0)],\upsilon)\leq h(\E [V^{\pi}_{0,\upsilon'}(S_0,0)],\upsilon');\notag
\end{equation}
    \item[(2)] for any $\gamma>0,\ \upsilon\in\Upsilon,$ there exists $\epsilon>0$ such that 
\begin{equation}
    h(\E [V^{\pi_{\upsilon,\epsilon}}_{0,\upsilon'}(S_0,0)],\upsilon')\leq h(\E [V^{\pi}_{0,\upsilon'}(S_0,0)],\upsilon')+\delta(\upsilon',\upsilon)+\gamma,\notag
\end{equation}
for any  $ \pi_{\upsilon,\epsilon}\in \left\{\pi\in\tilde{\Pi}:\E [V^{\pi}_{0,\upsilon}(S_0,0)]-\E [V^{*}_{0,\upsilon}(S_0,0)]<{\epsilon}\right\},$ 
\end{itemize} 
then we have 
\begin{equation}
   \lim_{n\to\infty,\epsilon\to0} h(\E [V^{\pi_{n,\epsilon}}_{0,\upsilon_{n+1}}(S_0,0) ],\upsilon_{n+1})= \min_{\upsilon\in\Upsilon}h(\E [V^{*}_{0,\upsilon}(S_0,0)],\upsilon).\notag
\end{equation}
\end{theorem}
\begin{remark}
    If Assumption C in \cite{bauerle2011markov} is satisfied, then there exists an optimal policy $\pi^*_{\upsilon}$ for any $\upsilon\in\Upsilon,$ which is defined in Remark \ref{remark_optimalpolicy}. As a result, condition (2) can be replaced by that
    \begin{equation}
    h(\E[ V^{\pi^*_{\upsilon}}_{0,\upsilon'}(S_0,0)],\upsilon')\leq h(\E [V^{\pi}_{0,\upsilon'}(S_0,0)],\upsilon')+\delta(\upsilon',\upsilon),\notag
\end{equation} for any $\upsilon,\upsilon'\in\Upsilon.$
\end{remark}

We summarize the complete Actor-Critic algorithm in Algorithm \ref{alg:alg1}.

\begin{algorithm}[h]
	\caption{Complete Recipe based on a Convex Scoring Function} 
	\label{alg:alg1}
	\begin{algorithmic}
        \STATE \textbf{Input:} ANNs for the value function $V^{\phi}$ and the policy $\pi^{\theta},$
        the number of episodes $N,$  epochs $K,$ the batch size $B,$ simulations $M,$ updating frequency $L,$ variance for exploration $\sigma^2.$
        \STATE Set initial guesses for $V^{\phi},\pi ^{\theta};$
        \STATE Initialize the environment and optimizers; 
        \STATE Initialize the buffer $\mathcal{M} = \emptyset$; 
        \STATE Initialize auxiliary variable $\upsilon^*=0;$
        \STATE Initialize counter $c=0;$
            \FOR {each epoch $\kappa=1$ to $K$} 
                \FOR {each episode $n=1$ to $N$}
                \STATE Sample initial state $s_0^{(n)}$ from initial distribution $S_0;$
                \STATE Sample $\upsilon^{(n)}$ from normal distribution $\mathcal{N}(\upsilon^*,\sigma^2);$
                \STATE Simulate episodes from $(s=s^{(n)}_0,y=0,\upsilon=y^{(n)}_0)$ under $\pi^{\theta};$ 
                \STATE Put every $\left(s^{(t,n)},y^{(t,n)},a^{(t,n)},s^{(t+1,n)},y^{(t+1,n)},\upsilon^{(n)}\right)$ into $\mathcal{M};$
                \ENDFOR
                \STATE Freeze $\tilde{\pi} = \pi^{\theta};$
                \STATE Estimate $V^{\phi}$ using $\tilde{\pi}:$
                \STATE \qquad Algorithm \ref{alg:alg2} $(V^{\phi},\tilde{\pi},\mathcal{M},K,B);$
                \STATE Freeze $\tilde{V} = V^{\phi};$ 
                \STATE Update $\pi^{\theta}$ using $\tilde{V}:$
                \STATE \qquad Algorithm \ref{alg:alg3} $(\tilde{V},{\pi}^{\theta},\mathcal{M},K,B);$
                \STATE Freeze $\tilde{V} = V^{\phi}$ and $\tilde{\pi} = \pi^{\theta};$                  \IF{$c\text{ mod }L=0$}
                \STATE Update $\upsilon^*$ using $\tilde{V}:$
                \STATE \qquad Algorithm \ref{alg:alg4} $(\tilde{V},K,M);$
                \STATE Reset $c=0;$
                \STATE Decrease $L$ and $\sigma^2;$
                \ENDIF
                \STATE Add counter $c$ by $1;$
                \STATE Clear buffer $\mathcal{M} = \emptyset$;  
            \ENDFOR
            \STATE \textbf{Output:} $\pi^{\theta}$ and $V^{\phi}.$ 
	\end{algorithmic}
\end{algorithm}

\subsubsection{Expected Shortfall (ES) Case}
$\ES$ is widely used in risk-sensitive RL (see \cite{bauerle2011markov}, \cite{chow2015risk}, \cite{tamar2015optimizing}, \cite{kashima2007risk}, \cite{du2022provably}, \cite{godbout2021acrel}). Our approach can handle the $\ES$ objective by taking $f(y,\upsilon)=\upsilon+(y-\upsilon)^+$ and $h(x,\upsilon)=x.$ We present here that under the $\ES$ objective, the sampling method described above can be further simplified to avoid the iterative use of SGD to solve \eqref{ystar}. For any initial $\upsilon^*_0$ and parameters $\phi_0$ and $\theta_0,$ we sample $\upsilon$ and update the parameters according to the following process: in stage $n\ (n\geq 1),$ we 
\begin{equation}
    \begin{aligned}
    & \text{Sample } \upsilon \sim \mathrm{N} \left(\upsilon^*_{n-1},\sigma_n^2\right)\text{ and generate trajectories}; \notag\\
    &  \text{Perform $L(n)$ Critic-steps and get $\phi_n$};\notag\\
    &  \text{Perform $L(n)$ Actor-steps and get $\theta_n$};\notag\\
    &  \text{Calculate } \upsilon^*_n \text{ as the }\alpha\text{-quantile of simulated total costs}.\notag
\end{aligned}
\end{equation}

Before discussing the theoretical convergence of this method, we introduce two notations. Given a policy $\pi\in\tilde{\Pi}$, we denote 
\begin{equation}
    F_{\pi}(y) := \mathbb{P}_{\pi}\left.\left(\sum_{t = 0}^{T-1} c(S_t, A_t,S_{t+1})\leq y\right|Y_0=0\right)\notag
\end{equation}
  and 
  \begin{equation}
      q_{\alpha}(F_{\pi}) := \inf\left\{y:F_{\pi}(y)\geq\alpha\right\}.\notag
  \end{equation}
  
\begin{proposition}\label{thm:limit}
Take $f(y,\upsilon) = \upsilon+\frac{1}{1-\alpha}(y-\upsilon)^+$ and $h(x,\upsilon)=x.$  For any $\epsilon>0,$ $\upsilon_1\in\Upsilon,$ define 
\begin{equation}
    \upsilon_{n+1} = q_{\alpha} (F_{\pi_{n,\epsilon}}).\notag
\end{equation}
It holds that $\lim_{n\to\infty}\E [V^{\pi_{n,\epsilon}}_{0,\upsilon_{n+1}}(S_0,0)]$ exists finitely. Furthermore, if  there exists some $\Delta>0$ such that for any $\upsilon'\in\Upsilon$ and $\pi\in\tilde{\Pi},$ the following two conditions hold:
\begin{itemize}
    \item[(1)] for any $\upsilon\in\arg\min_{\upsilon\in\Upsilon}\E [V^{\pi}_{0,\upsilon}(S_0,0)],$ we have
\begin{equation}
    \left|F_{\pi}(\upsilon')-F_{\pi}(\upsilon)\right| \geq (1-\alpha)\Delta\left|\upsilon'-\upsilon\right|;\notag
\end{equation}
    \item[(2)] for any $\gamma>0,$ $\upsilon\in\Upsilon,$ there exists $\epsilon>0$ such that
\begin{equation}
    \E [V^{\pi_{\upsilon,\epsilon}}_{0,\upsilon'}(S_0,0)]\leq \E [V^{\pi}_{0,\upsilon'}(S_0,0)]+\frac{\Delta}{2}(\upsilon'-\upsilon)^2+\gamma,\notag
\end{equation} for any $ \pi_{\upsilon,\epsilon}\in \left\{\pi\in\tilde{\Pi}:\E [V^{\pi}_{0,\upsilon}(S_0,0)]-\E[ V^{*}_{0,\upsilon}(S_0,0)]<{\epsilon}\right\},$
\end{itemize} then we have 
\begin{equation}
   \lim_{n\to\infty,\epsilon\to0} \E [V^{\pi_{n,\epsilon}}_{0,\upsilon_{n+1}}(S_0,0)] = \min_{\upsilon\in\Upsilon}\E [V^{*}_{0,\upsilon}(S_0,0)].\notag
\end{equation}
\end{proposition}

Based on the above result, we can implement the AC algorithm for the ES objective similar to Algorithm \ref{alg:alg1}, where we can update the auxiliary variable $\upsilon^*$ as the $\alpha$-quantile of 
$$\left\{\sum^{T-1}_{t=0}c(s^{(t,n)},a^{(t,n)},s^{(t+1,n)}):1\leq n\leq N\right\}.
$$

\section{Application: Statistical Arbitrage}\label{sec:app}
In this section, we present some numerical experiments in an algorithmic trading setting. The agent starts each session with a random initial inventory $Q_0$ and aims to trade quantities of an asset at each period $t \in \mathcal{T}$. The asset's price process is denoted by $P_t$. At each period, the agent observes the asset's price $P_t =  p_t \in \mathbb{R}$ and the inventory $Q_t =  q_t \in (-q_{\text{max}}, q_{\text{max}})$ and chooses a trading strategy $ A_t = a_t \in (-a_{\text{max}}, a_{\text{max}})$. The wealth process $ \left\{X_t\right\}_{t=0,1.\ldots,T}$ is governed by the MDP process:
$$
\begin{cases}
X_0 = 0, \\
X_t = X_{t-1} + Q_{t} (P_t-P_{t-1}) - \varphi A_{t-1}^2, & t = 1, \ldots, T-1, \\
X_T = X_{T-1} + Q_{T} (P_T-P_{T-1}) - \varphi A_{T-1}^2  - \psi Q_T^2,
\end{cases}
$$
with coefficients $ \varphi = 0.005$ and $\psi = 0.5$ representing the transaction costs and the terminal penalty incurred in trading, respectively. The cost is defined by $\text{cost}_t = X_{t} - X_{t+1}.$ Here, we choose $ T = 5 $, $ q_{\text{max}} = 5 $ and $ a_{\text{max}} = 2$. 
Let us assume that the asset price follows an Ornstein-Uhlenbeck process:
$$
\d P_t = \kappa (\mu- P_t) \d t + \sigma \d W_t,
$$
with $ \kappa = 2 $, $ \mu= 1 $, $ \sigma = 0.2 $, and $\{W_t\}_{0\leq t \leq T} $ is a standard Brownian motion. The initial distribution of $P_0$ is normally distributed:
$$
P_0 \sim \mathcal{N}\left(\mu,\frac{8\sigma^2}{{\kappa}}\right).
$$
The initial distribution of $Q_0$ follows a uniform distribution on $(-q_{\text{max}}, q_{\text{max}})$ and is independent with $P_0.$ The risk-sensitive agent aims to optimize the risk measure
\begin{equation}
\label{taarget}
    \quad \inf_{\pi \in \Pi} \rho \left(  \sum_{t = 0}^{T-1}\text{cost}_t \right).
\end{equation}
Similar to the discussion in \cite{cheridito2009time}, the dynamic optimization problem in \eqref{taarget} is time-inconsistent in general. However, our proposed method is applicable to the static ES objective, in contrast to the method proposed by \cite{CJ2024}, which has a time-consistent recursive risk measure as the optimization objective. 

We implement and compare several RL algorithms with scoring functions/risk measures from the following choices:
\begin{itemize}
    \item RL-mean: a conventional RL model minimizing the expectation of total costs; 
    \item RL-$\mathrm{ES}_{0.8}$\&$\mathrm{ES}_{0.6}:$ our proposed RL model minimizing the $\ES_{\alpha}$ $(\alpha=0.8, \; 0.6)$ value of total costs, i.e. $\rho(Y) = \E [Y|Y>q_{\alpha}(Y)]$; 
    \item RL-$\mathrm{Var}:$ our proposed RL model minimizing the variance of total costs, i.e. $\rho(Y) = \E[(Y-\E [Y])^2]$;
    \item     RL-$\mathrm{Mean}$-$\mathrm{Var}:$ our proposed RL model minimizing the mean-variance utility, i.e. $\rho(Y)=\E [Y] + \lambda \E[(Y-\E [Y])^2],$ where $\lambda $ is chosen as $1$;
    \item RL-OneStep$\mathrm{ES}_{0.8}$\&$\mathrm{ES}_{0.6}:$ the RL model proposed by \cite{CJ2024} with $\ES_{\alpha}$ $(\alpha=0.8, \; 0.6)$ one-step conditional risk measure, i.e. the RL problem is \begin{equation}
    \inf_{\pi\in\Pi} \rho\left(\text{cost}_0+\rho\left(\text{cost}_1+\rho(\cdots+\rho(\text{cost}_{T-1}))\right)\right),
    \end{equation} where $\rho$ is taken $\mathrm{ES}_{0.8}$ and $\mathrm{ES}_{0.6}.$ 
\end{itemize}

\subsection{The Approximation of the Value Function}
Figure \ref{value_function} shows the approximation of the value function in the ES case after several iterations in the Critic-step of training with the learned policy. The subplots correspond to states $s=(0.5,0),(1.0,0)$ and $(1.5,0)$ at time $0,$ respectively. The vertical axis is the value of  $V_{0,\upsilon}(s,0;\theta)$ or $V^{\phi}_{0,\upsilon}(s,0;\theta),$ while the horizontal axis corresponds to a different $\upsilon$-value. Here, the risk measure is $\mathrm{ES}_{0.8}.$ The blue line is the true value estimated by out-of-sample numerical simulation, while the orange line is the approximation given by the critic network. It can be seen that the critic network learns the value function of a given policy efficiently, which is in agreement with Theorem \ref{t1}.
\begin{figure}[h]
    \centering
    \begin{subfigure}{0.32\textwidth}
        \centering
        \includegraphics[width=\textwidth]{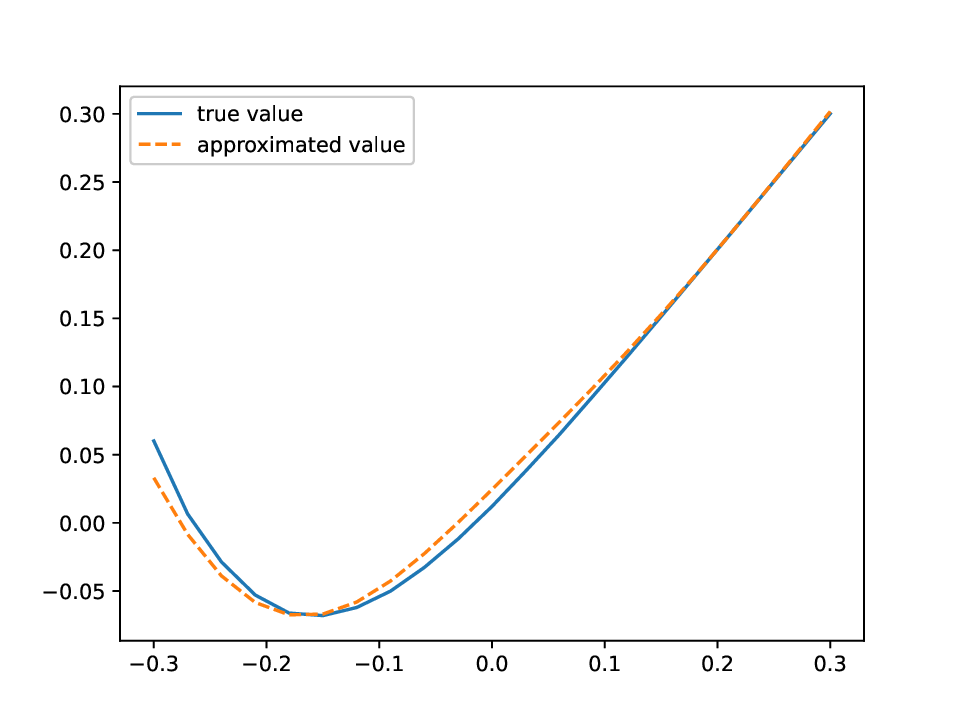}
        \caption{$t=0, \; s=(0.8,0)$}
    \end{subfigure}
    \begin{subfigure}{0.32\textwidth}
        \centering
        \includegraphics[width=\textwidth]{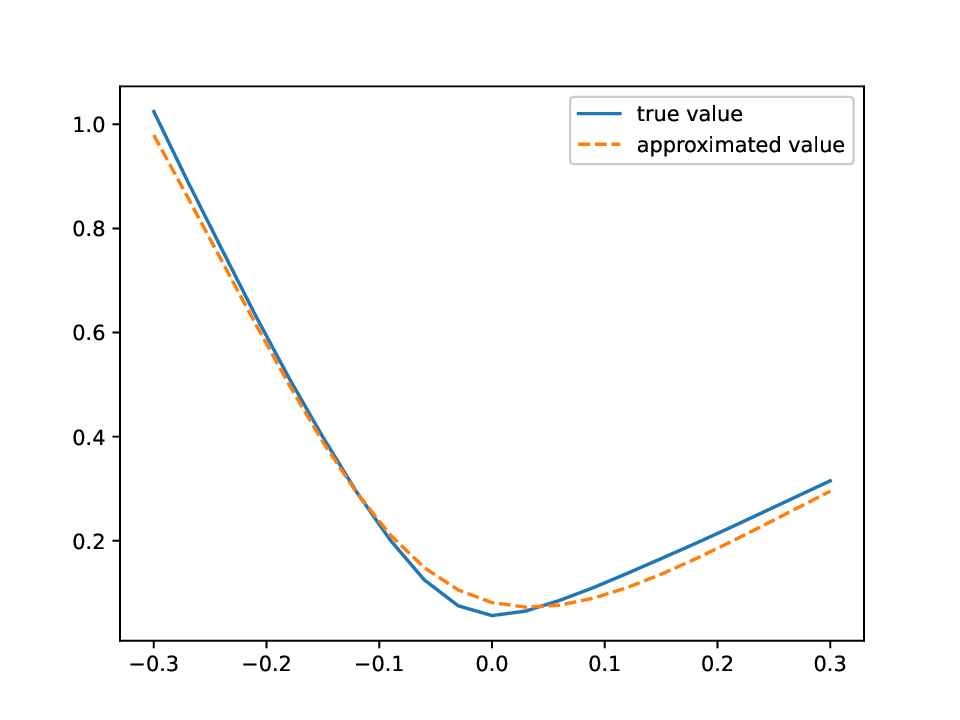}
        \caption{$t=0, \; s=(1.0,0)$}
    \end{subfigure}
\begin{subfigure}{0.32\textwidth}
        \centering
        \includegraphics[width=\textwidth]{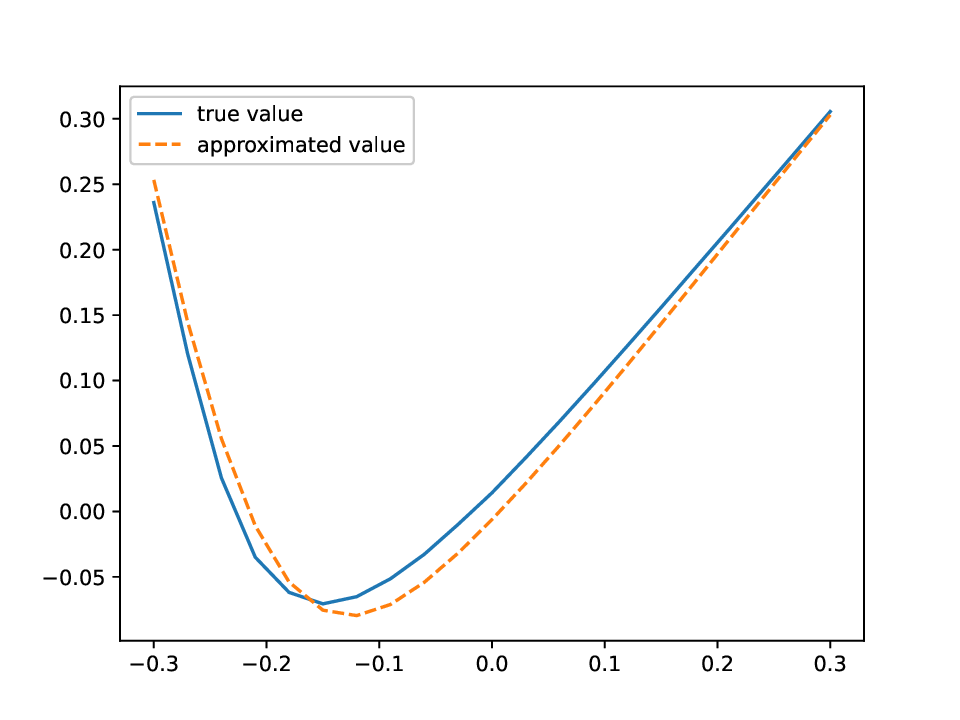}
        \caption{$t=0, \;s=(1.2,0)$}
    \end{subfigure}

    \caption{Approximation of the value function}
    \label{value_function}
\end{figure}

\subsection{Policy Interpretation (Initial and Dynamic)} 

Figure \ref{policyT0} shows a comparison of the learned policy among the five models. In each subplot, the horizontal axis represents the inventory $Q,$ and the vertical axis represents the price $S.$ In general, the agent tends to sell when the price is high and buy when the price is low; it also tends to sell when the inventory is high and buy when the inventory is low. We can distinguish the behaviors of different models by the shape of the buy-sell boundary. RL-mean is the most aggressive, as the buy-sell boundary is almost flat, indicating that the agent engages in trading when the price deviates slightly from the mean price $S=1.$ The behavior of the RL-$\mathrm{ES}_{\alpha}$ agents is more similar to that of the conventional RL-mean model, except that their buy-sell boundary is steeper, indicating that as their inventory accumulates, the threshold for the signal to continue trading increases.

\begin{figure}[h!]
    \centering
    \begin{minipage}{0.3\textwidth}
        \centering
        \includegraphics[width=\textwidth]{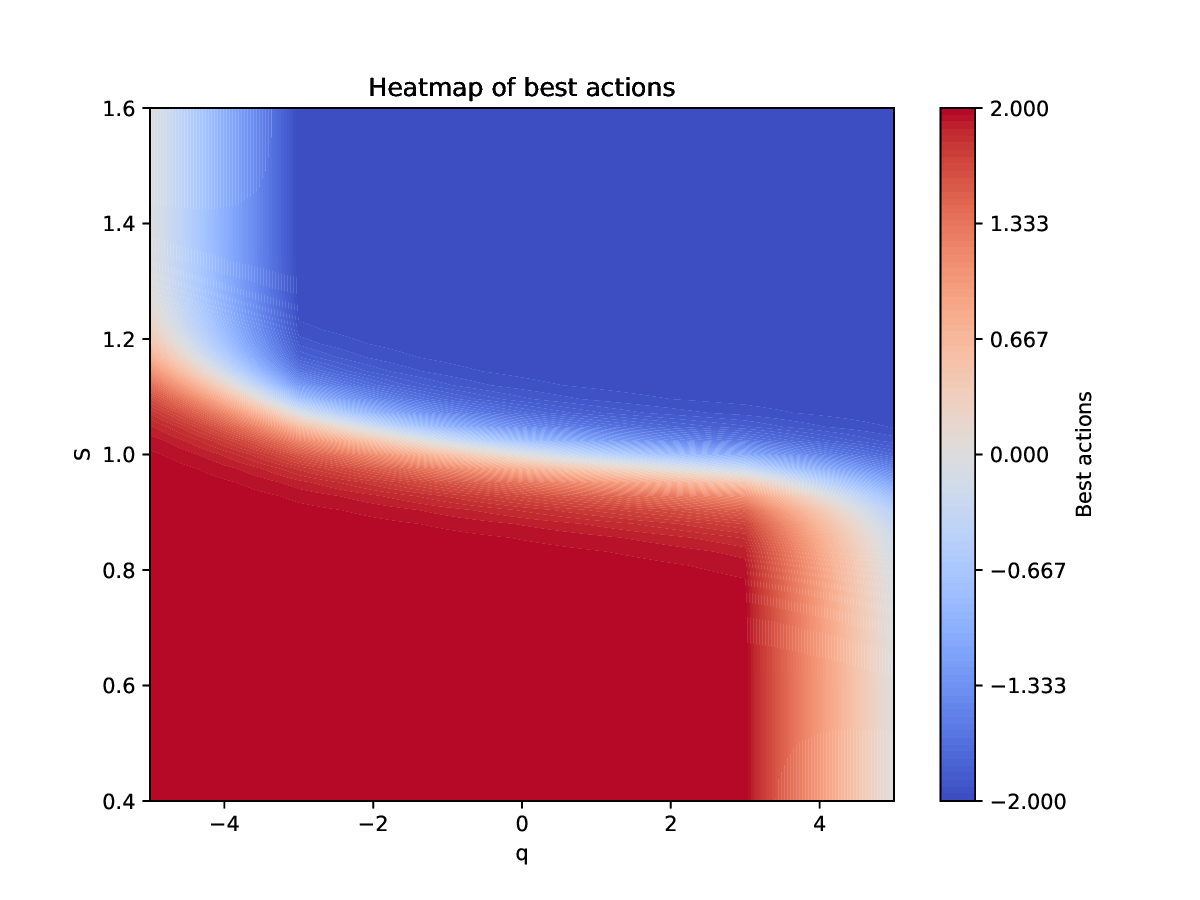}
        \subcaption{RL-mean}\label{fig:image1}
    \end{minipage}\hfill
    \begin{minipage}{0.3\textwidth}
        \centering
        \includegraphics[width=\textwidth]{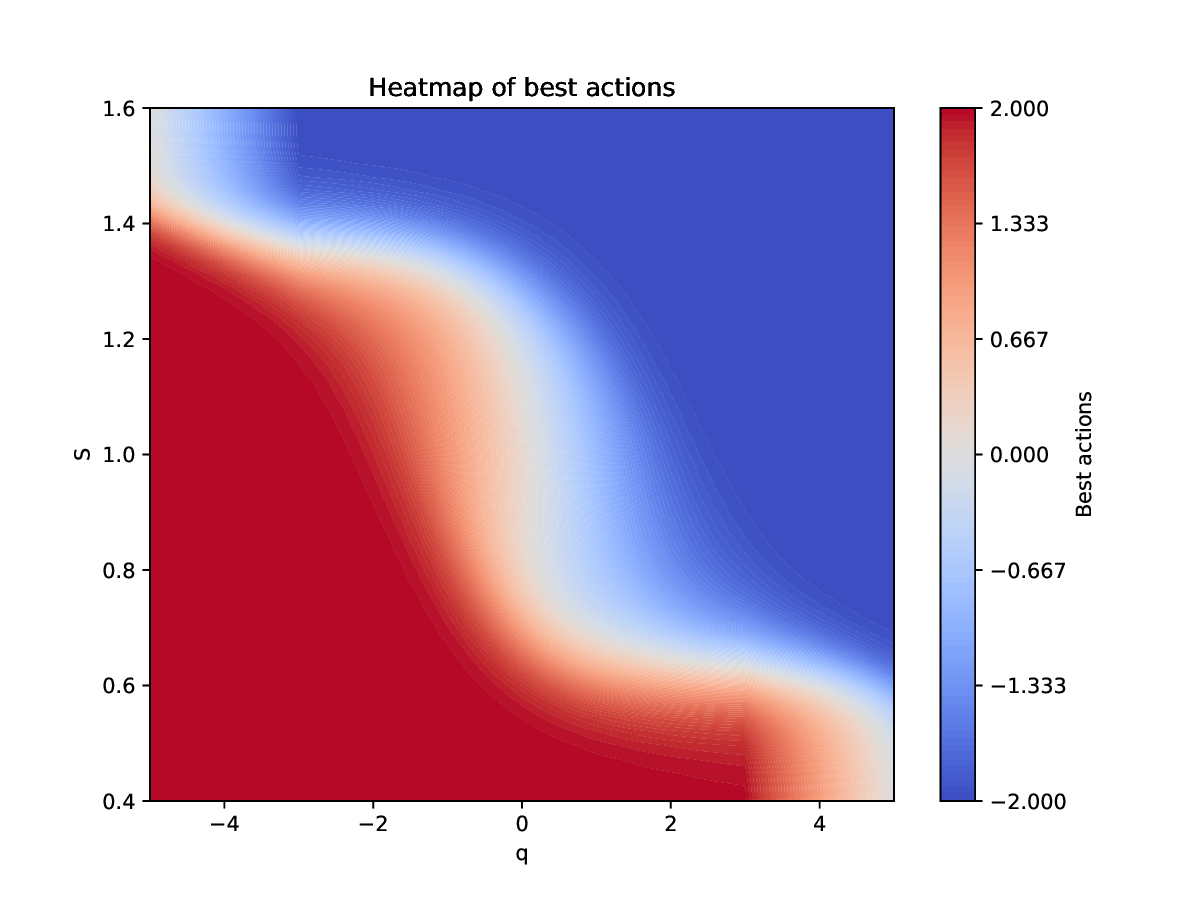}
        \subcaption{RL-OneStep$\mathrm{ES}_{0.8}$}\label{fig:image2}
    \end{minipage}\hfill
    \begin{minipage}{0.3\textwidth}
        \centering
        \includegraphics[width=\textwidth]{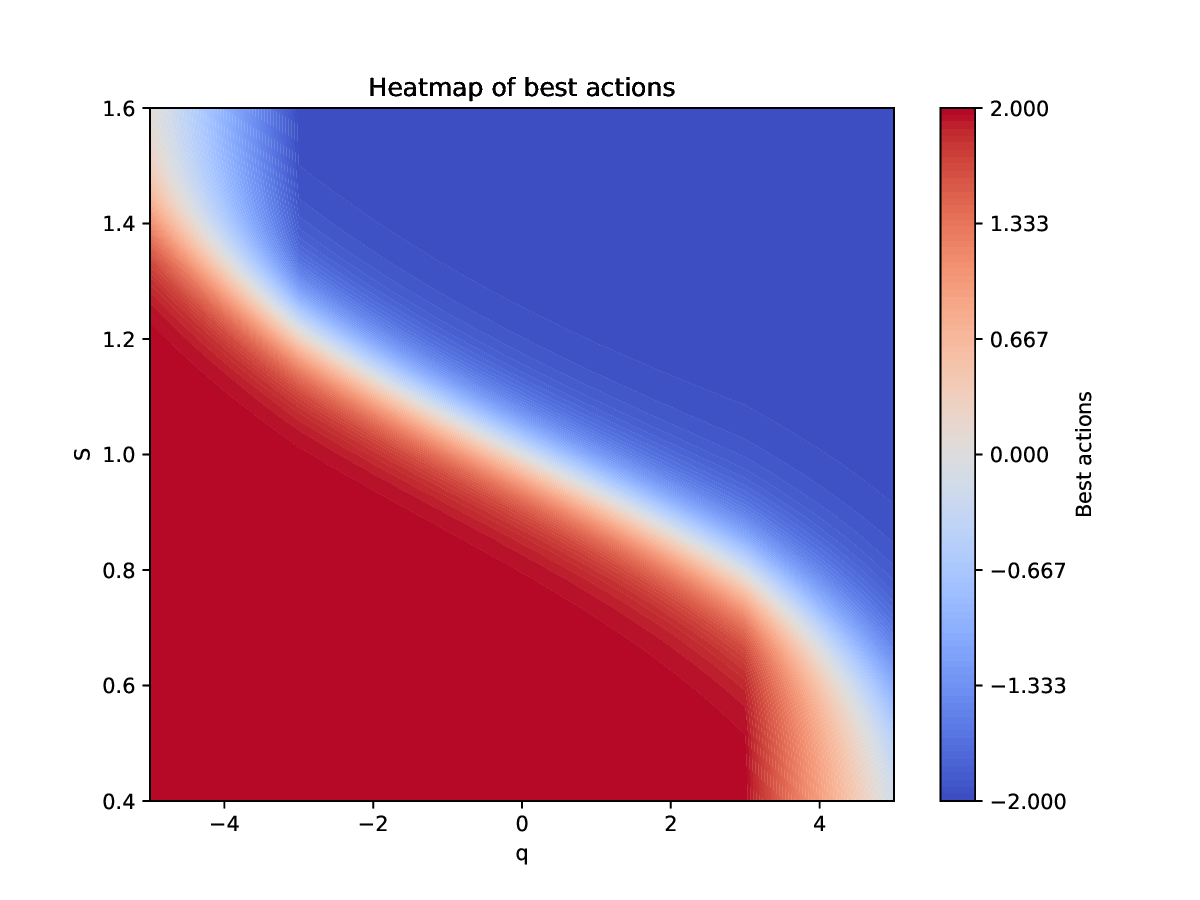}
        \subcaption{RL-$\mathrm{ES}_{0.8}$}\label{fig:image3}
    \end{minipage}

    \vspace{1em}

    \begin{minipage}{0.3\textwidth}
        \centering
        \includegraphics[width=\textwidth]{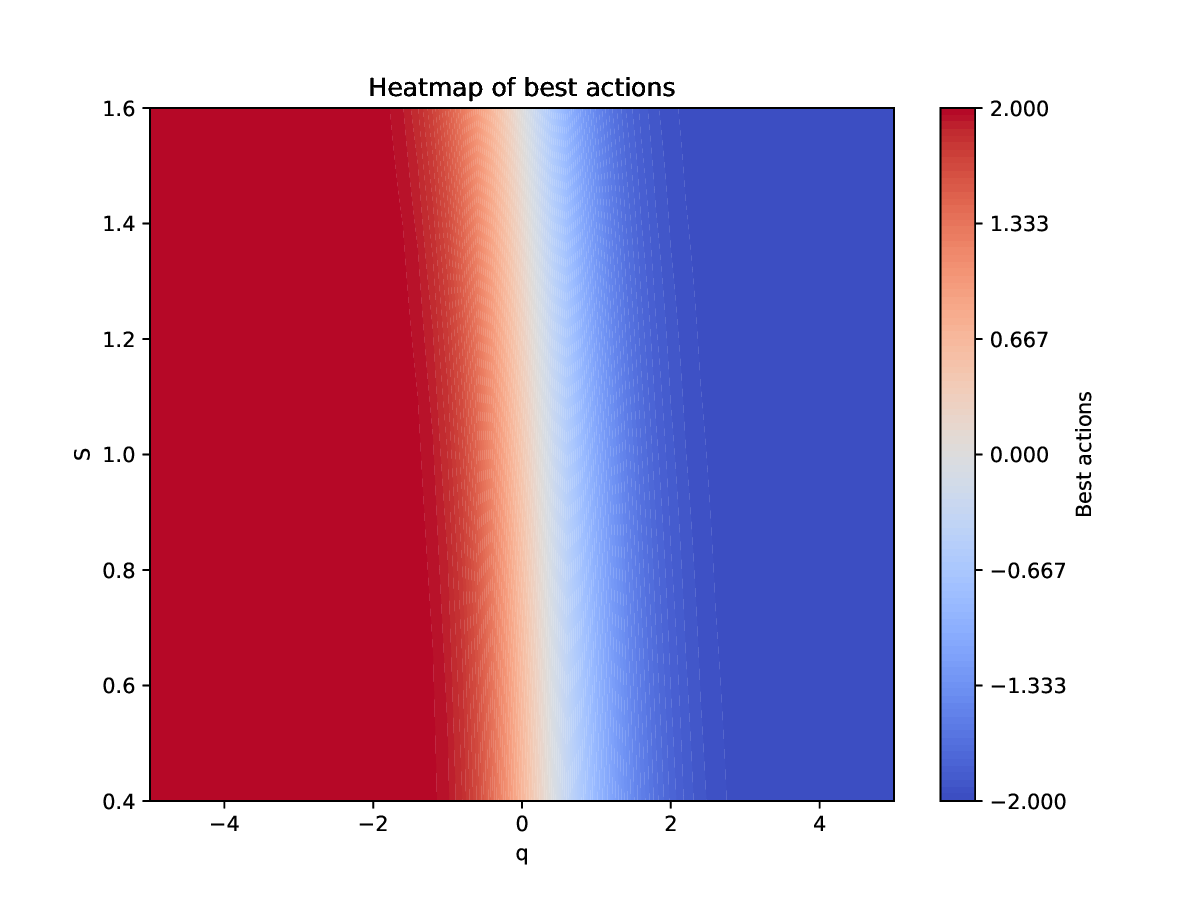}
        \subcaption{RL-$\mathrm{Var}$}\label{fig:image4}
    \end{minipage}\hspace{0.05\textwidth}
    \begin{minipage}{0.3\textwidth}
        \centering
        \includegraphics[width=\textwidth]{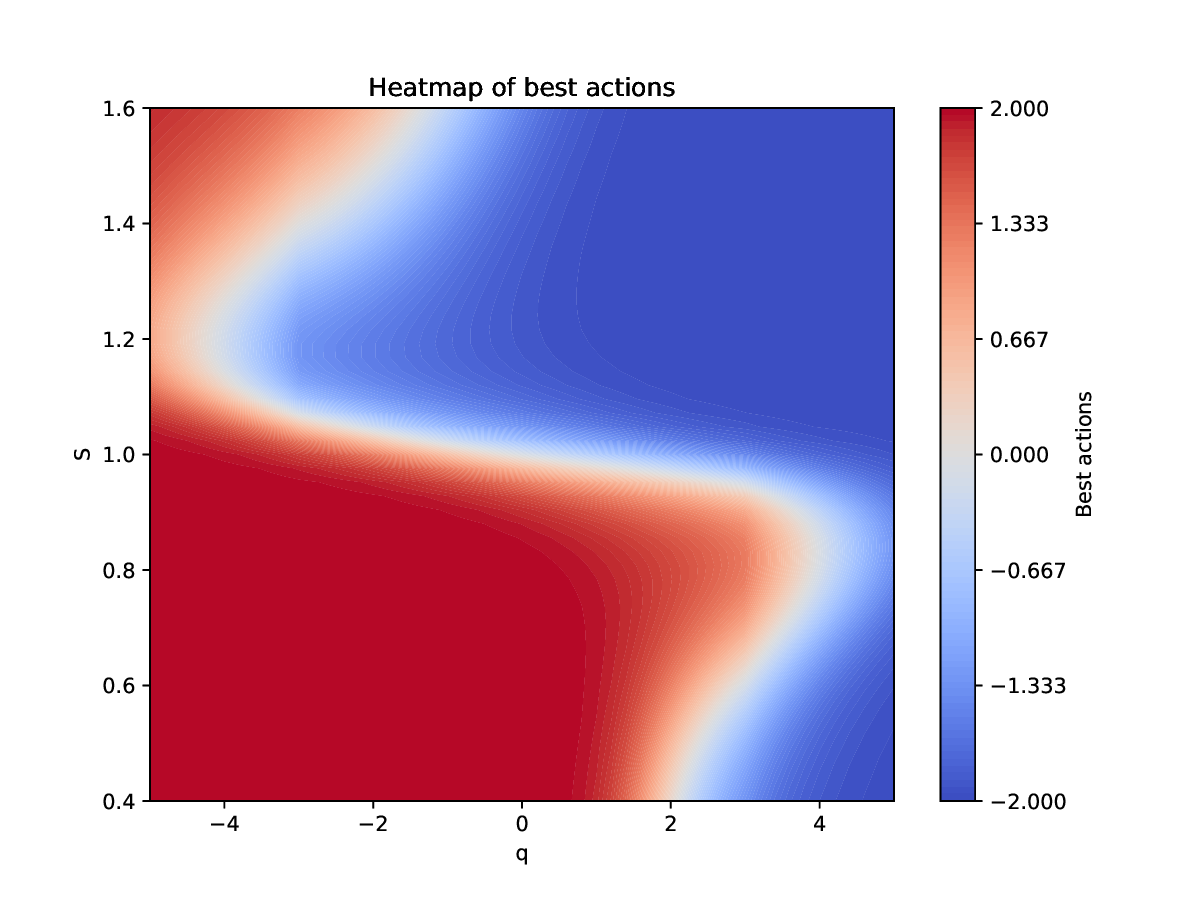}
        \subcaption{RL-$\mathrm{Mean}$-$\mathrm{Var}$}\label{fig:image5}
    \end{minipage}
    \caption{Learned policy at time $0$}
    \label{policyT0}
\end{figure}

Figure \ref{polict_after0} shows the policy learned after time $0.$ Since the proposed method treats $y$ as a new (augmented) state, to understand the learned strategy, we propose a three-dimensional heat map illustration, which represents the action on different $(y, S, q)$ triples. Similar to time $0,$ the agent still suggests buying when the price is low and selling when the price is high again, but the policy becomes conservative as the value of $y$ increases. This can be illustrated with an example at $t=2$ in Figure \ref{polict_after0}, noting that when the price is very low (close to 0.5), with the same inventory $Q=2,$ the agent buys when the $y$ value is smaller and tends to maintain the position when the $y$ value is larger. Another point is that by comparing $t=1$ and $t=4$, one can see that the policy becomes conservative over time. At the time $4$, the action is hardly affected by the price and just tries to close the position as much as possible, most likely due to the terminal-time penalty. 

\begin{figure}[h!]
    \centering
    \begin{minipage}{0.45\textwidth}
        \centering
        \includegraphics[width=\textwidth]{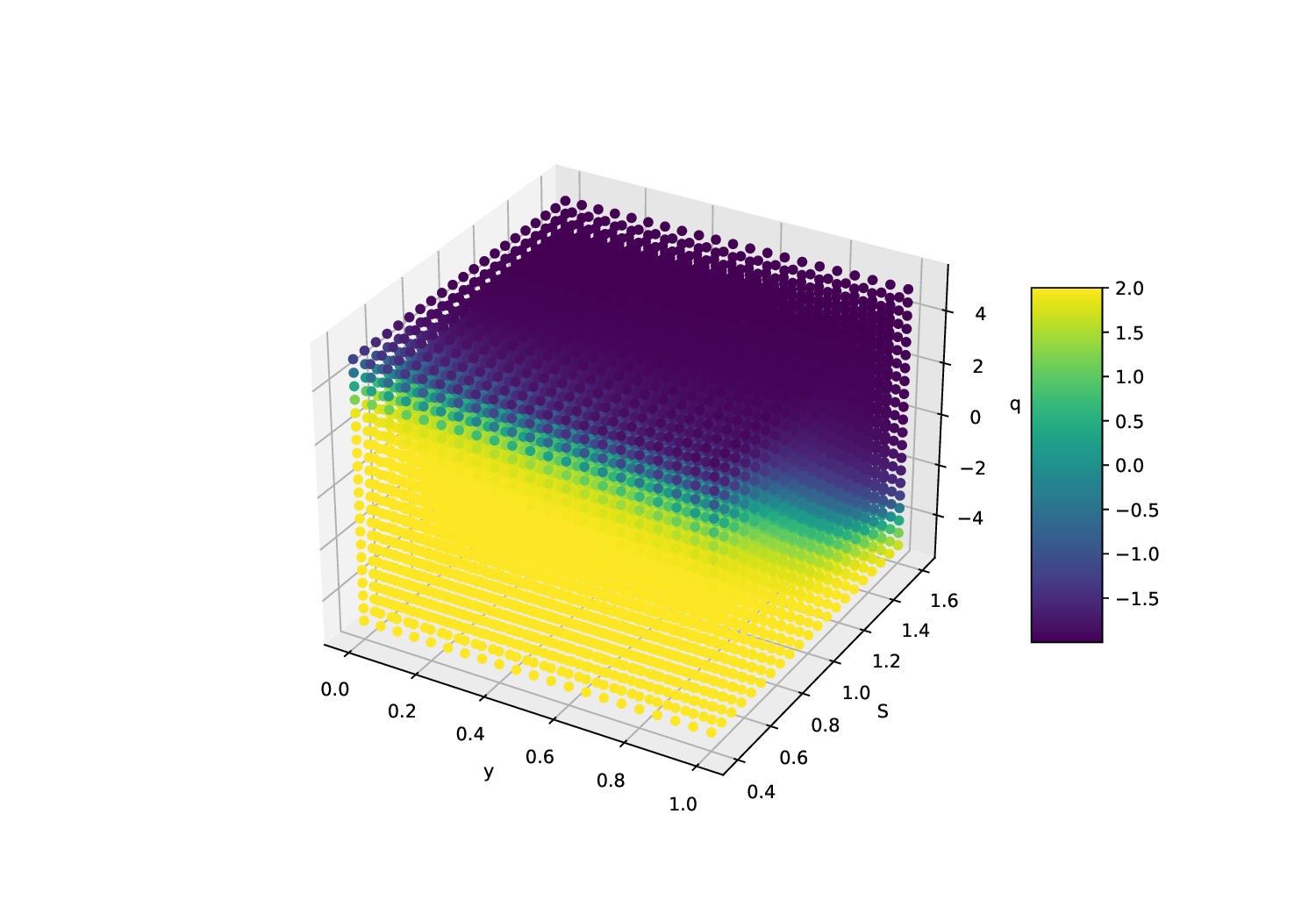}
        \subcaption{$t=1$}\label{fig:image1}
    \end{minipage}\hspace{0.05\textwidth} 
    \begin{minipage}{0.45\textwidth}
        \centering
        \includegraphics[width=\textwidth]{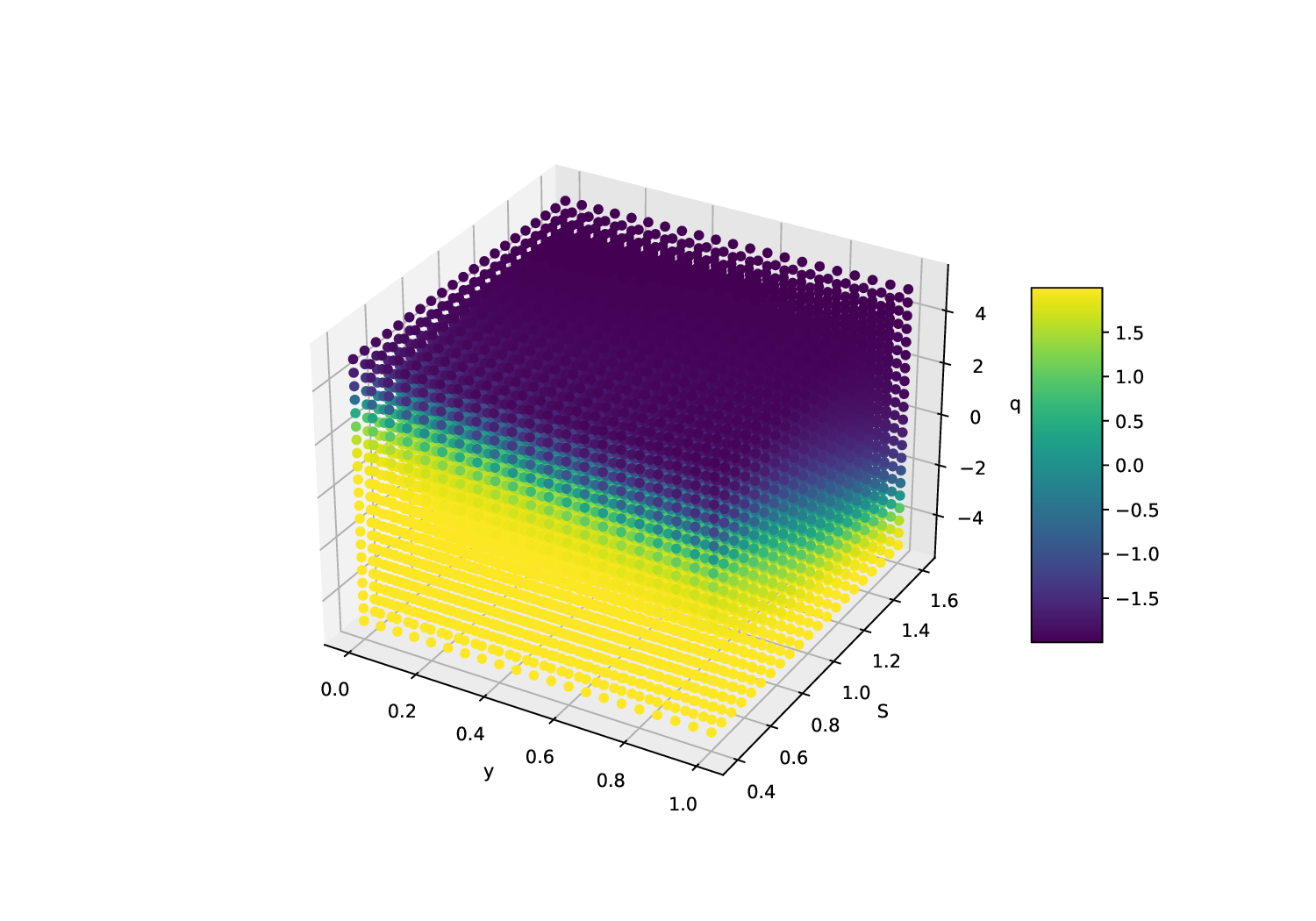}
        \subcaption{$t=2$}\label{fig:image2}
    \end{minipage}

    \vspace{1em} 

    \begin{minipage}{0.45\textwidth}
        \centering
        \includegraphics[width=\textwidth]{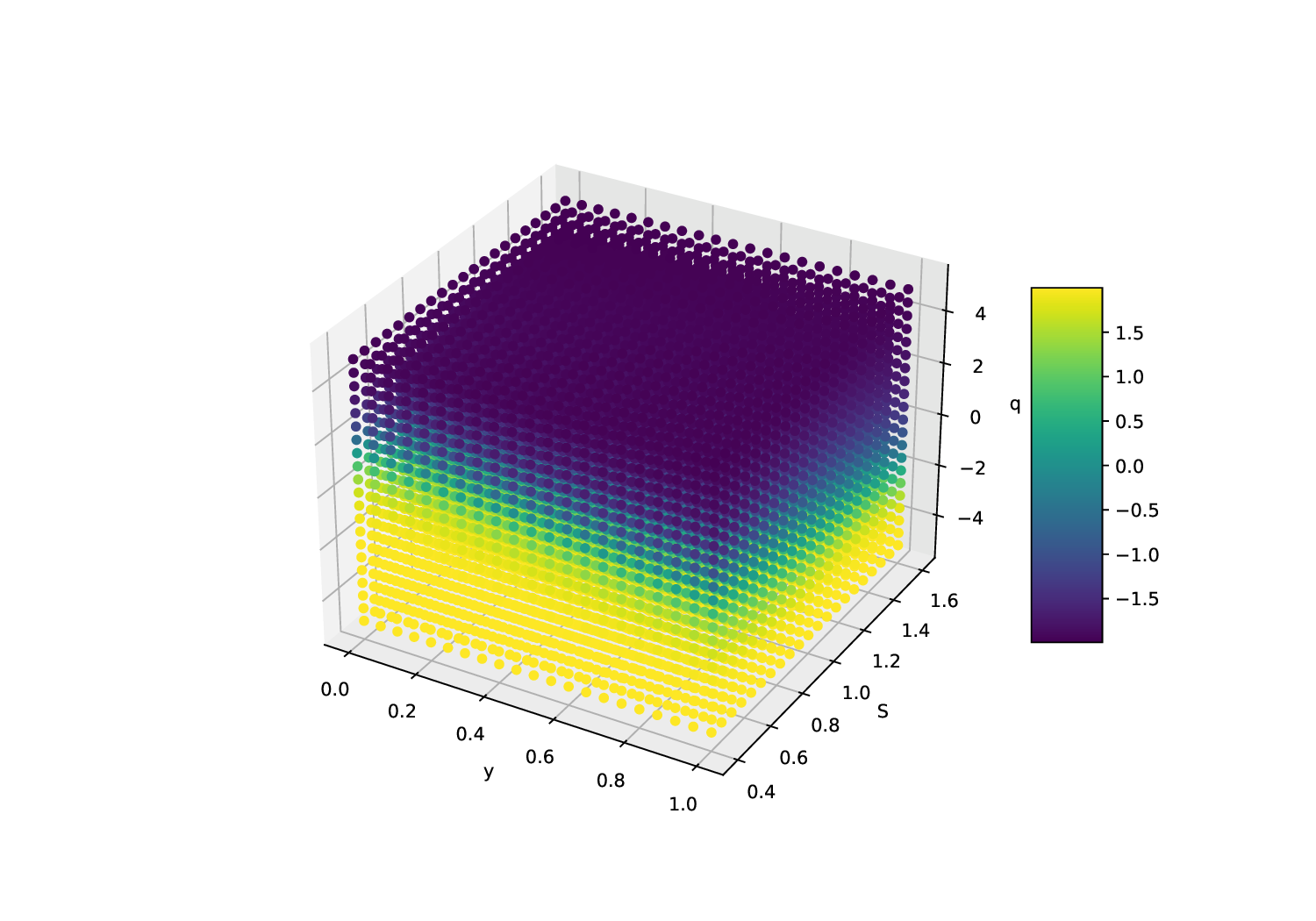}
        \subcaption{$t=3$}\label{fig:image3}
    \end{minipage}\hspace{0.05\textwidth}
    \begin{minipage}{0.45\textwidth}
        \centering
        \includegraphics[width=\textwidth]{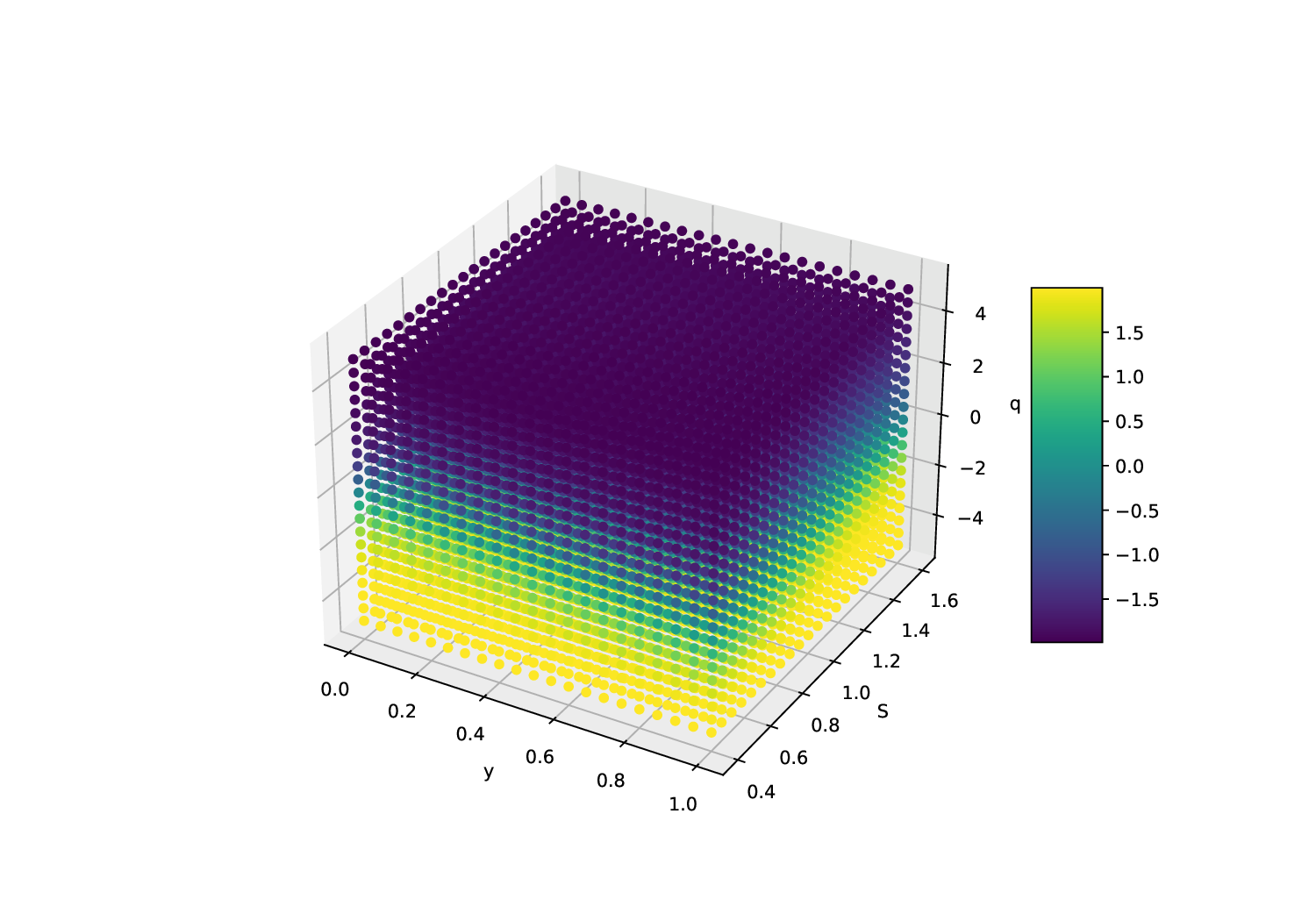}
        \subcaption{$t=4$}\label{fig:image4}
    \end{minipage}
    \caption{Learned policy as a function of time (RL-$\mathrm{ES}_{0.8}$ )}
    \label{polict_after0}
\end{figure}

\subsection{Performance comparisons}

We analyze the total costs' distribution under the policy of each model based on 300,000 out-of-sample simulations. The results are shown in Table \ref{table1}, where the value of the statistics of the total costs is shown. RL-mean achieves the lowest mean cost (-0.524), but its $\mathrm{ES}_{0.8}$ (0.253) and $\mathrm{ES}_{0.6}$ (0.086) values are not necessarily the lowest, reflecting a trade-off between optimizing for the mean and risk measures. RL-$\mathrm{ES}_{0.8}$ achieves the lowest $\mathrm{ES}_{0.8}$ of the total costs (0.215), which aligns with its optimization objective. RL-$\mathrm{ES}_{0.6}$ shows the lowest $\mathrm{ES}_{0.6}$ of the total costs (0.079), reflecting its focus on mitigating risks in the worst 40\% of outcomes. Our proposed method better controls the Expected Shortfall value compared to the other two methods. Under different $\alpha$'s, the learned policies achieve the lowest $\mathrm{ES}_{\alpha}$ values, respectively. The method we propose has a clear advantage in addressing \eqref{taarget} and serves as a valuable tool for controlling tail risk.

\begin{table}
\centering
\caption{Summary Statistics of Out-sample Total Costs}
\label{table1}
\begin{tabular}{lccccc}
\toprule
    & Mean    & $\mathrm{ES}_{0.8}$  & $\mathrm{ES}_{0.6}$& $\mathrm{Var}$ &$\mathrm{Mean}$-$\mathrm{Var}$ utility\\
\midrule
RL-mean          & -\textbf{0.524}& 0.253 & 0.086 & 0.523& -0.002\\
RL-OneStep$\mathrm{ES}_{0.8}$ & -0.256 & 0.251 & 0.129 & 0.335& 0.080\\
RL-OneStep$\mathrm{ES}_{0.6}$ & -0.422 & 0.224 & 0.101 & 0.463& 0.040\\
RL-$\mathrm{ES}_{0.8}$        & -0.261 & \textbf{0.215}& 0.083 & 0.149& -0.113\\
RL-$\mathrm{ES}_{0.6}$        & -0.395 & 0.225 & \textbf{0.079}& 0.301& -0.092\\
RL-$\mathrm{Var}$        & 0.094& 0.331& 0.232& \textbf{0.030}& 0.124\\
RL-$\mathrm{Mean}$-$\mathrm{Var}$        & -0.272& 0.246& 0.081& 0.122& -\textbf{0.150}\\
\bottomrule
\end{tabular}
\end{table}

\begin{figure}
    \centering
    \includegraphics[width=0.5\linewidth]{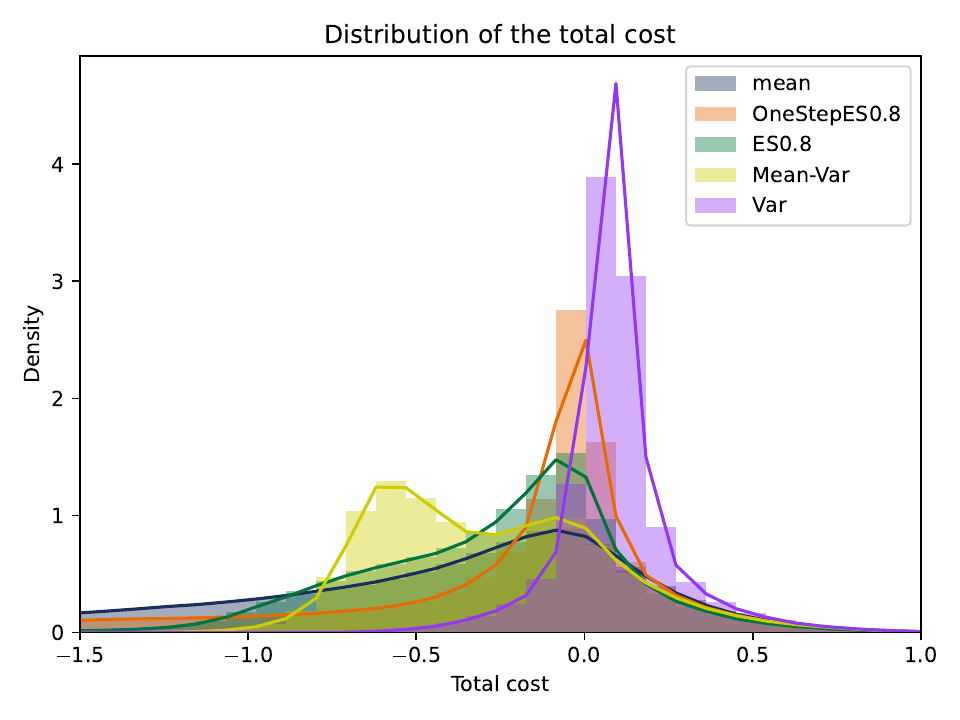}
    \caption{Distribution of the total costs}
    \label{fig:distribution}
\end{figure}

\begin{figure}
    \centering
    \begin{subfigure}{0.45\textwidth}
        \centering
        \includegraphics[width=\textwidth]{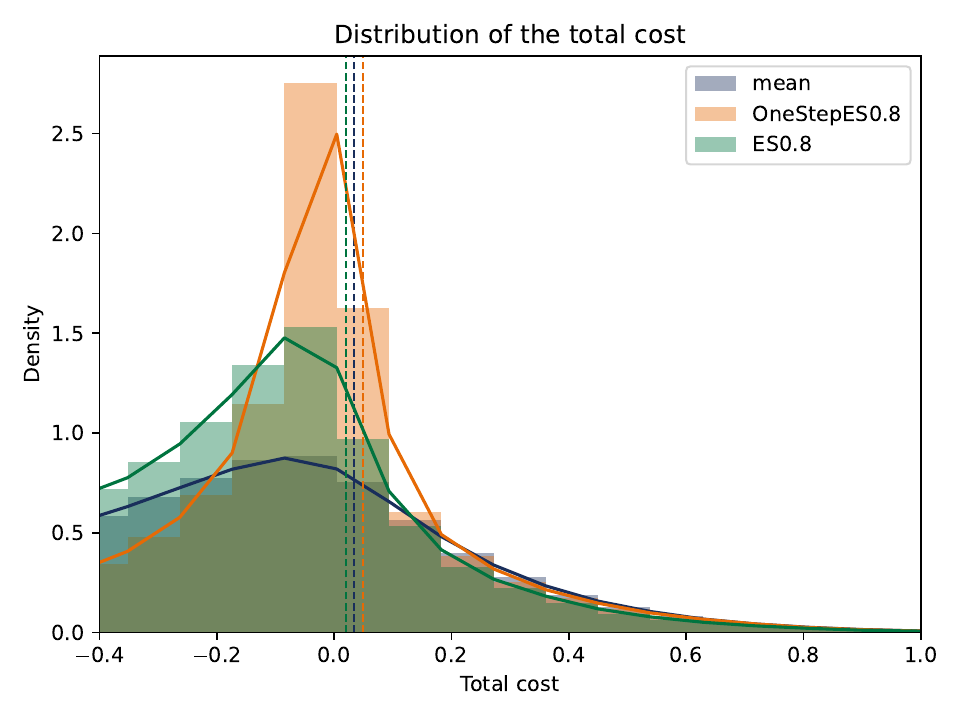}
        \caption{0.8 quantile group} 
    \end{subfigure}
    \hspace{0.05\textwidth} 
    \begin{subfigure}{0.45\textwidth}
        \centering
        \includegraphics[width=\textwidth]{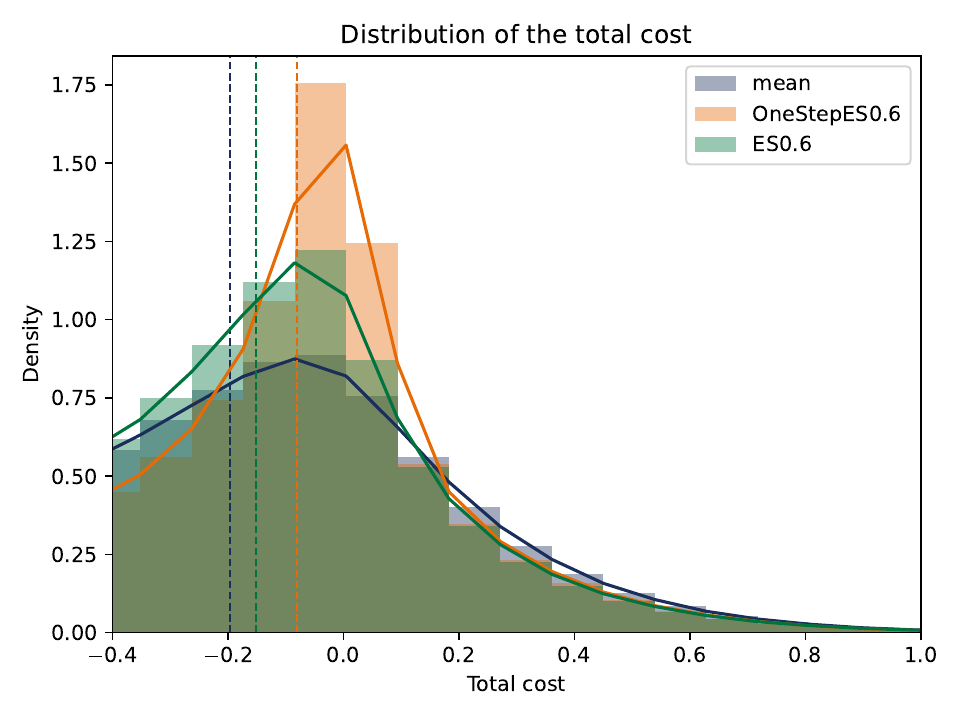}
        \caption{0.6 quantile group} 
    \end{subfigure}
    \caption{Tail distribution of the total costs} 
    \label{fig:tail_dis_com}
\end{figure}

Figure \ref{fig:distribution} shows a comparison of estimated distributions of the total costs between the five models. In general, the total costs of all RL-OneStep$\mathrm{ES}_{\alpha}$ models and RL-$\mathrm{ES}_{\alpha}$ models have a smaller variance compared to the result of the RL-mean model, which is due to the risk-sensitive objectives. Though the results of our approach show a greater variance than the results of RL-OneStep$\mathrm{ES}_{\alpha},$  it leads to a thinner tail distribution where the total costs are large. This can be seen more clearly in Figure \ref{fig:tail_dis_com}, in which the two subplots correspond to $\alpha=0.8$ and $\alpha=0.6$ and are compared with RL-mean, respectively. The dashed lines mark the $\alpha$-quantiles of the total costs distributions. 

In conclusion, we hope that our contributions may open a new research avenue by integrating the scoring functions into some risk-sensitive RL tasks, and the proposed algorithm in the present paper may help to motivate more diverse and efficient risk-sensitive RL algorithms.

\section{Proofs}\label{sec:proof}
This section collects all proofs of the results in the previous sections.

\begin{proof}[Proof of Proposition \ref{simpleProposition2}]
We first show that $\upsilon\mapsto g(h(\E [f(Y,\upsilon)],\upsilon))$ is convex. For any $\lambda\in(0,1)$ and $\upsilon_1,\upsilon_2\in\Upsilon$, note that
\begin{align}
   &g( h(\E [f(Y,\lambda\upsilon_1+(1-\lambda)\upsilon_2)],\lambda\upsilon_1+(1-\lambda)\upsilon_2)))\notag\\
   \leq &g( h(\lambda\E [f(Y,\upsilon_1)]+(1-\lambda)\E [f(Y,\upsilon_2)],\lambda\upsilon_1+(1-\lambda)\upsilon_2)))\notag\\
   \leq & \lambda g(h(\E [f(Y,\upsilon_1)],\upsilon_1))+ (1-\lambda) g(h(\E [f(Y,\upsilon_2)],\upsilon_2)).\notag
\end{align}
The first inequality is due to the fact that $\E[ f(Y,\cdot)]$ is convex by Proposition \ref{simpleProposition1}, and the second inequality is because of the convexity of $g(h(\cdot,\cdot))$. As a result, for any $y\in[\underline{y},\bar{y}],$ $g(h( f(y,\upsilon),\upsilon))$ is convex in $\upsilon.$ Moreover, 
\begin{equation}
\label{bound_neq}
g\left(h\left(\min_{y\in[\underline{y},\bar{y}]}f(y,\upsilon),\upsilon\right)\right) \leq g(h(\E[ f(Y,\upsilon)],\upsilon)).
\end{equation}

\textbf{Case 1:}
\begin{equation}
\lim_{\upsilon\to \pm\infty}g\left(h\left(\min_{y\in[\underline{y},\bar{y}]}f(y,\upsilon),\upsilon\right)\right) = \infty.\notag
\end{equation}
In this case, by \eqref{bound_neq} we have
\begin{equation}
     \lim_{\upsilon \to \infty}g(h(\E [f(Y,\upsilon)],\upsilon))=\lim_{\upsilon \to -\infty}g(h(\E [f(Y,\upsilon)],\upsilon))=\infty,\notag
\end{equation}
which implies that $\upsilon\mapsto g(h(\E[ f(Y,\upsilon)],\upsilon))$ has a minimizer. Suppose that the claim in Proposition \ref{simpleProposition2} does not hold. Then there exist random variables $Y_1,Y_2,\ldots,$ and 
real numbers $\upsilon_1,\upsilon_2,\ldots,$ such that: 
\begin{itemize}
\item either $\upsilon_n\to\infty \ (n\to\infty)$ and $\upsilon\mapsto g(h(\E [f(Y_n,\upsilon)],\upsilon))$ is strictly decreasing on $(-\infty,\upsilon_n]$;
\item or $\upsilon_n\to-\infty \ (n\to\infty)$ and $\upsilon\mapsto g(h(\E [f(Y_n,\upsilon)],\upsilon))$ is strictly increasing on $[\upsilon_n,\infty).$  
\end{itemize}
Here we only discuss the first case, as the other case can be handled similarly. Then, we have
\begin{align}
g\left(h\left(\max_{y\in[\underline{y},\bar{y}]}f(y,\upsilon_0),\upsilon_0\right)\right) &\geq g(h(\E [f(Y_n,\upsilon_0)],\upsilon_0))\notag\\
    &\geq g(h(\E [f(Y_n,\upsilon_n)],\upsilon_n))\notag\\
    &\geq g\left(h\left(\min_{y\in[\underline{y},\bar{y}]}f(y,\upsilon_n),\upsilon_n\right)\right),\notag
\end{align}
which implies
\begin{equation}
g\left(h\left(\max_{y\in[\underline{y},\bar{y}]}f(y,\upsilon_0),\upsilon_0\right)\right)\geq \lim_{n\to\infty} g\left(h\left(\min_{y\in[\underline{y},\bar{y}]}f(y,\upsilon_n),\upsilon_n\right)\right) = \infty,\notag
\end{equation}
leading to a contradiction.

\textbf{Case 2:}  $h(\cdot,\cdot)$ is constant with respect to its second argument, i.e., $h(x,\upsilon)=\tilde{h}(x),$ and $f(y,\upsilon) =w(y)+\phi(\upsilon-\psi(y)),$ for some continuous functions $w,\phi,\psi.$ Let 
\begin{align}
\bar{\upsilon}^*     &= \sup_{y\in[\underline{y},\bar{y}]}\left\{\inf\left\{\upsilon\in\mathbb{R}:\tilde{h}(f(y,\upsilon)) = \min_{\upsilon\in\mathbb{R}}\tilde{h}(f(y,\upsilon))\right\}\right\},\notag\\
     \underline{\upsilon}^* &= \inf_{y\in[\underline{y},\bar{y}]}\left\{\sup\left\{\upsilon\in\mathbb{R}:\tilde{h}(f(y,\upsilon)) = \min_{\upsilon\in\mathbb{R}}\tilde{h}(f(y,\upsilon))\right\}\right\}.\notag
\end{align}
Then $\underline{\upsilon}^*,\bar{\upsilon}^*\in\{\infty,-\infty\}\cup\mathbb{R}.$ As $h$ is strictly increasing in the first argument, $\tilde{h}$ must be strictly increasing. 
Then $\upsilon\mapsto f(y,\upsilon)$ and $\upsilon\mapsto\tilde{h}(f(y,\upsilon))$ are strictly decreasing on $( -\infty,\bar{\upsilon}^*),$ and strictly increasing on $( \underline{\upsilon}^*,\infty)$ for any $y\in[\underline{y},\bar{y}].$ 
Thus $\upsilon\mapsto\E [f(Y,\upsilon)]$ is also strictly decreasing on $( -\infty,\bar{\upsilon}^*)$ and strictly increasing on $( \underline{\upsilon}^*,\infty).$ 
Note that $\upsilon\mapsto f(y,\upsilon)$ is convex. We also have that $\upsilon\mapsto f(y,\upsilon)$ is increasing on $(\bar{\upsilon}^*,\infty),$ and decreasing on $( -\infty,\underline{\upsilon}^*),$ so is $\upsilon\mapsto\E [f(Y,\upsilon)].$
We now show that $\bar{\upsilon}^*<\infty$ and $\underline{\upsilon}^*>-\infty.$ We only need to prove $\bar{\upsilon}^*<\infty$ and the other is similar. Assume that $\bar{\upsilon}^*=\infty.$ Then there exist $y_1,y_2,\ldots,$ and $\upsilon_1,\upsilon_2,\ldots,$ such that 
\begin{equation}
   \upsilon_n =  \inf\left\{\upsilon\in\mathbb{R}:\tilde{h}(f(y_n,\upsilon)) = \min_{\upsilon\in\mathbb{R}}\tilde{h}(f(y_n,\upsilon))\right\}>-\infty,\notag
\end{equation}
and $\upsilon_n\to\infty\ (n\to\infty).$ Note that
\begin{equation}
    \upsilon_n -\psi(y_n) =\inf\{\arg\min_{x\in\mathbb{R}} \phi(x)\},\notag
\end{equation}
which implies $\upsilon_n-\psi(y_n)=\upsilon_1-\psi(y_1).$ Thus we have
\begin{align}
    |\upsilon_n| &\leq |\upsilon_1| + |\psi(y_1)|+|\psi(y_n)|\notag\\
    &\leq  |\upsilon_1| + 2\max_{y\in[\underline{y},\bar{y}]} |\psi(y)|,\notag
\end{align}
which leads to a contradiction with $\upsilon_n$ tending to infinity. Thus, we have $\bar{\upsilon}^*<\infty$ and $\underline{\upsilon}^*>-\infty.$ As $\upsilon\mapsto\E [f(Y,\upsilon)]$ is convex and $\upsilon^*$ is a minimizer of $\upsilon\mapsto \E [f(Y,\upsilon)]$ iff it is a minimizer of $\upsilon\mapsto\tilde{h}(\E [f(Y,\upsilon)]),$ we can define an interval $[m^{\underline{y},\bar{y}},M^{\underline{y},\bar{y}}]$ as follows: If $\bar{\upsilon}^*,\underline{\upsilon}^*\in\mathbb{R},$ then we take $m^{\underline{y},\bar{y}}<\min\{\bar{\upsilon}^*,\underline{\upsilon}^*\},$ $M^{\underline{y},\bar{y}}>\max\{\bar{\upsilon}^*,\underline{\upsilon}^*\}$ and the conclusion follows; If $\bar{\upsilon}^*=-\infty$ and $\underline{\upsilon}^*=\infty,$ then $\tilde{h}(\E [f(Y,\upsilon)])$ is constant of $\upsilon$ on $\mathbb{R}$ and the conclusion follows on any sub-interval of $\mathbb{R};$ If $\bar{\upsilon}^*=-\infty$ and $\underline{\upsilon}^*\in\mathbb{R},$  then $\tilde{h}(\E [f(Y,\upsilon)])$ is constant of $\upsilon$ on $(-\infty,\underline{\upsilon}^*]$ and is increasing on $[\underline{\upsilon}^*,\infty),$ and we take $m^{\underline{y},\bar{y}}<\underline{\upsilon}^*,$ $M^{\underline{y},\bar{y}}>\underline{\upsilon}^*$;  If $\bar{\upsilon}^*\in\mathbb{R}$ and $\underline{\upsilon}^*=\infty,$ then $\tilde{h}(\E [f(Y,\upsilon)])$ is constant of $\upsilon$ on $[\bar{\upsilon}^*,\infty)$ and is increasing on $(-\infty,\bar{\upsilon}^*]$, and we take $m^{\underline{y},\bar{y}}<\bar{\upsilon}^*,$ $M^{\underline{y},\bar{y}}>\bar{\upsilon}^*.$ 
\end{proof}

\begin{proof}[Proof of Theorem \ref{t_Bellman}]
Define $V_{T,\upsilon}^*(s,y) = f(-y,\upsilon),$ and define $V_{t,\upsilon}^*(s,y) = \mathcal{J}_t V^*_{t+1,\upsilon}(s,y)$ recursively for $t=T-1.\ldots,1,0.$ Below, we use the backward induction to show that $V_{t,\upsilon}^*(s,y) \leq  V^{\pi}_{t,\upsilon}(s,y)$ for any $s \in \mathcal{S},$  $y \in \mathcal{Y}$ and $t=0,\ldots,T.$ At time $T,$ $V_{T,\upsilon}^*(s,y) = V^{\pi}_{T,\upsilon}(s,y) = f(-y,\upsilon).$ Assuming the assertion holds for $t+1,$ we obtain at time $t,$
\begin{align}
    V^{\pi}_{t,\upsilon}(s,y) &= \mathbb{E}_{\pi} \left.\left[V^{\pi}_{t+1,\upsilon}\left(S_{t+1},y-c\left(s,A_t,S_{t+1}\right)\right)\right|S_t=s,Y_t=y\right]\notag\\
    &= \int_{\mathcal{A}} \int_{\mathcal{S}} \pi(a|t,s,y) p(s'|s,a)V^{\pi}_{t+1,\upsilon}\left(s',y-c\left(s,a,s'\right)\right) \nu_{\mathcal{S}}(\d s') \nu_{\mathcal{A}}(\d a)\notag\\
    &\geq \int_{\mathcal{A}} \int_{\mathcal{S}} \pi(a|t,s,y) p(s'|s,a)V^*_{t+1,\upsilon}\left(s',y-c\left(s,a,s'\right)\right) \nu_{\mathcal{S}}(\d s') \nu_{\mathcal{A}}(\d a)\notag\\
    & \geq \int_{\mathcal{A}}  \pi(a|t,s,y)\left( \inf_{a'\in\mathcal{A}}\int_{\mathcal{S}} p(s'|s,a')V^*_{t+1,\upsilon}\left(s',y-c\left(s,a',s'\right)\right) \nu_{\mathcal{S}}(\d s') \right)\nu_{\mathcal{A}}(\d a) \notag\\
    &= \inf_{a\in\mathcal{A}}\int_{\mathcal{S}} p(s'|s,a)V^*_{t+1,\upsilon}\left(s',y-c\left(s,a,s'\right)\right) \nu_{\mathcal{S}}(\d s')\notag\\
    &= V_{t,\upsilon}^*(s,y).\notag
\end{align}
By induction, the assertion holds for $0\leq t\leq T.$

For any $(s,y)\in\mathcal{S}\times\mathcal{Y},$ there exists a sequence $\{a_n\}_{n=1}^{\infty}\subset\mathcal{A}$ such that 
\begin{align}
& \int_{\mathcal{S}}  p(s'|s,a_n)f\left(c\left(s,a_n,s'\right)-y,\upsilon\right) \nu_{\mathcal{S}}(\d s')\notag\\
&\quad\to \inf_{a\in\mathcal{A}}\left\{\int_{\mathcal{S}}  p(s'|s,a)f\left(c\left(s,a,s'\right)-y,\upsilon\right) \nu_{\mathcal{S}}(\d s')\right\}.\notag
\end{align}
We hence let $\pi_{\upsilon,n} (\cdot|T-1,s,y)$ be a point mass on $a_n.$ Moreover, for any $(s,y)\in\mathcal{S}\times\mathcal{Y},$ there exists a sequence $\{a_n\}_{n=1}^{\infty}\in\mathcal{A}$ such that 
\begin{align}
&\int_{\mathcal{S}}  p(s'|s,a_n){V}^*_{t+1,\upsilon}\left(s',y-c\left(s,a_n,s'\right)\right) \nu_{\mathcal{S}}(\d s')\notag\\
&\quad\to   \inf_{a\in\mathcal{A}}\left\{\int_{\mathcal{S}}  p(s'|s,a){V}^*_{t+1,\upsilon}\left(s',y-c\left(s,a,s'\right)\right) \nu_{\mathcal{S}}(\d s')\right\}.\notag
\end{align}
We define $\pi_{\upsilon,n} (\cdot|t,s,y)$ recursively as a point mass on $a_n.$ By the dominated convergence theorem, it is easy to see that
\begin{equation}
  \left[ V_{t,\upsilon}^{\pi_{\upsilon,n}}(s,y) -V_{t,\upsilon}^*(s,y)\right]\to 0 \ (n\to\infty),\notag
\end{equation}
for any $t\in\mathcal{T},s\in\mathcal{S},y\in\mathcal{Y}.$
\end{proof}

The proof of Proposition \ref{propmin} needs the following lemma.

\begin{lemma}
\label{lpxs_lemma}
There exists a $\Delta>0$ such that 
\begin{equation}
    \left|\E[f(Y,\upsilon_1)]-\E [f(Y,\upsilon_2)]\right| \leq \Delta |\upsilon_1-\upsilon_2|,\notag
\end{equation}
for any $\upsilon_1,\upsilon_2\in\Upsilon$ and random variable $Y$ satisfying $\underline{y}\leq Y\leq\bar{y}$ almost surely. 
\end{lemma}
\begin{proof}[Proof of Lemma \ref{lpxs_lemma}]
For any $y\in\mathcal{Y},$ and $\upsilon_1<\upsilon_2\in\Upsilon,$ let
\begin{equation}
    L(y,\upsilon_1,\upsilon_2) := \frac{f(y,\upsilon_2)-f(y,\upsilon_1)}{\upsilon_2-\upsilon_1}.\notag
\end{equation}
As $f(y,\upsilon)$ is convex in $\upsilon\in\mathbb{R},$ $L$ is increasing in $\upsilon_2$ and $\upsilon_1$ respectively. Hence, it is sufficient to show 
\begin{equation}
\label{leftside}\inf_{y\in\mathcal{Y},\upsilon_2\downarrow\underline{\upsilon}}L(y,\underline{\upsilon},\upsilon_2) >-\infty,\\
\end{equation} and \begin{equation}\label{rightside}\sup_{y\in\mathcal{Y},\upsilon_1\uparrow \bar{\upsilon}}L(y,\upsilon_1,\bar{\upsilon}) <\infty,
\end{equation}
which then imply
\begin{equation}
    -\infty<\inf_{y\in\mathcal{Y},\upsilon_1<\upsilon_2\in\Upsilon} \frac{f(y,\upsilon_2)-f(y,\upsilon_1)}{\upsilon_2-\upsilon_1}\leq\sup_{y\in\mathcal{Y},\upsilon_1<\upsilon_2\in\Upsilon} \frac{f(y,\upsilon_2)-f(y,\upsilon_1)}{\upsilon_2-\upsilon_1} <\infty.\notag
\end{equation}
In what follows, we only show \eqref{rightside}, and the proof for \eqref{leftside} is similar. We prove by contradiction and suppose that \eqref{rightside} does not hold. Then there exist $\left\{y_n\right\}^{\infty}_{n=1}\subset \mathcal{Y}$ and $\upsilon_n\uparrow\bar{\upsilon}\ (n\to\infty)$ such that 
$f(y_n,\bar{\upsilon})-f(y_n,\upsilon_n)>n(\bar{\upsilon}-\upsilon_n).$ 
Since $\underline{y}\leq y_n \leq \bar{y}$ for all $n\geq 1,$ there exists a convergent subsequence $y_{n_k}\to y_0 \; (k\to\infty).$ Thus, we have
\begin{equation}
\liminf_{k\to\infty}   \frac{ f(y_0,\bar{\upsilon})-f(y_0,\upsilon_{n_k})}{\bar{\upsilon}-\upsilon_{n_k}} =\liminf_{k\to\infty}   \frac{ f(y_{n_k},\bar{\upsilon})-f(y_{n_k},\upsilon_{n_k})}{\bar{\upsilon}-\upsilon_{n_k}} \geq  \liminf_{k\to\infty} n_k = \infty,\notag
\end{equation}
which implies 
\begin{equation}
    \lim_{k\to\infty}   \frac{ f(y_0,\bar{\upsilon})-f(y_0,\upsilon_{n_k})}{\bar{\upsilon}-\upsilon_{n_k}} = \infty.\notag
\end{equation}
This leads to a contradiction with the fact that $f(y_0,\cdot)$ is convex on $\mathbb{R}.$ Therefore, there exists $\Delta>0$ such that 
\begin{equation}
    \sup_{y\in\mathcal{Y},\upsilon_1<\upsilon_2\in\Upsilon} \frac{|f(y,\upsilon_2)-f(y,\upsilon_1)|}{|\upsilon_2-\upsilon_1|} \leq \Delta.\notag
\end{equation}
Finally, we have
\begin{align}
    \left|\E[ f(Y,\upsilon_1)]-\E [f(Y,\upsilon_2)]\right| = &\left|\int_{\mathcal{Y}}f(y,\upsilon_1) \d F(y)-\int_{\mathcal{Y}}f(y,\upsilon_2) \d F(y)\right|\notag\\
   \leq &\int_{\mathcal{Y}}\left|f(y,\upsilon_1) -f(y,\upsilon_2) \right|\d F(y)
   \leq \Delta|\upsilon_1-\upsilon_2|.\notag\qedhere
\end{align}
\end{proof}

\begin{proof}[Proof of Proposition \ref{propmin}]
First, we prove that $\upsilon\mapsto h(\mathbb E [V^*_{0,\upsilon}(S_0,0)] ,\upsilon)$ is a continuous function on $\Upsilon.$ We prove this claim by contradiction, and suppose that $\upsilon \mapsto h(\mathbb E [V^*_{0,\upsilon}(S_0,0)], \upsilon)$ has a discontinuous point at $\upsilon_0.$ Then, there exists a sequence $\left\{\upsilon_n\right\}_{n=1}^{\infty}$ satisfying 
     $\upsilon_n\to\upsilon_0\ (n\to\infty)$
and
\begin{equation}
\label{abs_gap}
    \left|\mathbb E [V^*_{0,\upsilon_n}(S_0,0)]-\mathbb E [V^*_{0,\upsilon_0}(S_0,0)]\right| \geq \delta
\end{equation}
for some $\delta>0.$ We note that 
\begin{align}
    &\mathbb E [V^*_{0,\upsilon_0}(S_0,0)]\notag\\
    =&\inf_{\pi\in\tilde{\Pi}}\lim_{\upsilon\to\upsilon_0}  \mathbb{E}_{\pi} \left[f\left(\sum_{k=0}^{T-1}c(S_k,A_k,S_{k+1})-Y_t,\upsilon\right)\right]\notag\\
    \geq&\lim_{\upsilon\to\upsilon_0}\inf_{\pi\in\tilde{\Pi}} \mathbb{E}_{\pi} \left[f\left(\sum_{k=0}^{T-1}c(S_k,A_k,S_{k+1})-Y_t,\upsilon\right)\right]\notag\\
    =&\lim_{\upsilon\to\upsilon_0}\mathbb E [V^*_{0,\upsilon}(S_0,0)],\notag
\end{align}
which implies that $\mathbb E [V^*_{0,\upsilon}(S_0,0)]$ is upper semi-continuous in $\upsilon.$ Therefore, we only need to show
\begin{equation}
    \mathbb E [V^*_{0,\upsilon_n}(S_0,0)]-\mathbb E [V^*_{0,\upsilon_0}(S_0,0)] \leq -\delta. \notag
\end{equation}
For any $\epsilon>0,$ take  
\begin{equation}
   \pi_{n,\epsilon}\in \left\{\pi\in\tilde{\Pi}:\E [V^{\pi}_{0,\upsilon_n}(S_0,0)]-\E [V^{*}_{0,\upsilon_n}(S_0,0)]<\epsilon\right\}. \notag
\end{equation}
Then, it holds that 
\begin{equation}
     \mathbb E [V^*_{0,\upsilon_0}(S_0,0)]\leq \liminf_{\epsilon\to 0} \E [V^{\pi_{n,\epsilon}}_{0,\upsilon_0}(S_0,0)],\notag
\end{equation}
which implies
\begin{equation}
  \limsup_{\epsilon\to 0} \E [V^{\pi_{n,\epsilon}}_{0,\upsilon_n}(S_0,0)]  =\mathbb E [V^*_{0,\upsilon_n}(S_0,0)]\leq \liminf_{\epsilon\to 0} \E [V^{\pi_{n,\epsilon}}_{0,\upsilon_0}(S_0,0)]-\delta.\notag
\end{equation}
Thus we have 
\begin{equation}
\label{eq_contra}
    \limsup_{\epsilon\to 0} \frac{\E [V^{\pi_{n,\epsilon}}_{0,\upsilon_n}(S_0,0)]-\E [V^{\pi_{n,\epsilon}}_{0,\upsilon_0}(S_0,0)]}{|\upsilon_n-\upsilon_0|}\leq -\frac{\delta}{|\upsilon_n-\upsilon_0|} .
\end{equation}
In \eqref{eq_contra}, as $n\to\infty,$ the $\text{RHS}\to-\infty$ while the $\text{LHS}\geq-\Delta$ by Lemma \ref{lpxs_lemma}. This leads to a contradiction.

As $\upsilon\mapsto h(\mathbb E [V^*_{0,\upsilon}(S_0,0) ],\upsilon)$ is a continuous function on $\Upsilon,$ it attains a minimal value $l$ on $\Upsilon.$ We next prove that $l$ is also the minimal value of $h(\mathbb E [V^*_{0,\upsilon}(S_0,0)] ,\upsilon)$ on $\mathbb{R}.$ Assume that there exists $\upsilon_0\notin \Upsilon,$ such that  $h(\mathbb E [V^*_{0,\upsilon_0}(S_0,0)] ,\upsilon_0)<l.$ By Proposition  \ref{prop_appro}
, there exists a $\pi\in\tilde{\Pi}$ such that $h(\mathbb E [V^{\pi}_{0,\upsilon_0}(S_0,0)] ,\upsilon_0)<l.$ Then by Lemma \ref{simpleProposition1}  and Proposition \ref{simpleProposition2}, $h(\mathbb E [V^{\pi}_{0,\upsilon}(S_0,0)] ,\upsilon)$ attains its minimal value on $\upsilon_1\in\Upsilon,$ and thus we have
\begin{equation}
   h(\mathbb E [V^{\pi}_{0,\upsilon_1}(S_0,0)] ,\upsilon_1)\leq h(\mathbb E [V^{\pi}_{0,\upsilon_0}(S_0,0)] ,\upsilon_0)<l,\notag
\end{equation}
which leads to a contradiction with the fact that $l$ is the minimal value on $\Upsilon.$
\end{proof}

\begin{proof}[Proof of Theorem \ref{t1}]
By Theorem 2.2 in \cite{hornik1989multilayer} and the fact that
\begin{align}
    V_{t,\upsilon}(s,y;\theta) &= \E_{\pi_{\upsilon}^\theta} \left[f\left(\sum_{k=t}^{T-1}c(S^{\theta}_k,A^{\theta}_k,S^{\theta}_{k+1})-Y_t^{\theta},\upsilon\right) \Big|S_t^{\theta}=s, Y_t^{\theta}=y\right]\notag\\
    &=\E_{\pi_{\upsilon}^\theta} \left[f\left(\sum_{k=t+1}^{T-1}c(S^{\theta}_k,A^{\theta}_k,S^{\theta}_{k+1})+c(s,A^{\theta}_t,S^{\theta}_{t+1})-y,\upsilon\right) \right]\notag
\end{align}
is measurable in $(s,y,\upsilon)$ and continuous in $(y,\upsilon),$ for any $\gamma_{t},\epsilon_{t}>0,$ there exists a neural network with parameter $\phi_t$ such that
\begin{equation}
         \nu_{\mathcal{S}}\left(\sup_{y\in\mathcal{Y},\upsilon\in\Upsilon}| V_{t,\upsilon}(s,y;\theta) -V^{\phi_t}_{t,\upsilon}(s,y;\theta)| > \gamma_{t}\right)< \epsilon_{t}.\notag
    \end{equation}
By Lemma 6.4 in \cite{CJ2024}, we know that, for any $\hat{\gamma}>0$, there exists a neural network $\phi$ such that
\begin{equation}
    \underset{t\in\mathcal{T},s\in S,y \in \mathcal{Y},\upsilon\in \Upsilon}{\sup}\  \left| V^{\phi_t}_{t,\upsilon}(s,y;\theta)-V^{\phi}_{t,\upsilon}(s,y;\theta)\right| < \hat{\gamma},\notag
\end{equation}
which implies
\begin{equation}
     \nu_{\mathcal{S}}\left(\sup_{t\in\mathcal{T},y\in\mathcal{Y},\upsilon\in \Upsilon}| V^{\phi_t}_{t,\upsilon}(s,y;\theta) -V^{\phi}_{t,\upsilon}(s,y;\theta)| > \hat{\gamma}\right) = 0.\notag
\end{equation}
Thus, we have 
\begin{align}
     &\nu_{\mathcal{S}}\left(\sup_{t\in\mathcal{T},y\in\mathcal{Y},\upsilon\in\Upsilon}| V_{t,\upsilon}(s,y;\theta) -V^{\phi}_{t,\upsilon}(s,y;\theta)| > \sum_{i=0}^{T-1}\gamma_{i}+\hat{\gamma}\right)\notag\\
     \leq & \nu_{\mathcal{S}}\left(\sup_{t\in\mathcal{T},y\in\mathcal{Y},\upsilon\in\Upsilon}| V_{t,\upsilon}(s,y;\theta) -V^{\phi}_{t,\upsilon}(s,y;\theta)| > \sum_{i=0}^{T-1}\gamma_{i}+\hat{\gamma},\sup_{t\in\mathcal{T},y\in\mathcal{Y},\upsilon\in\Upsilon}| V^{\phi_t}_{t,\upsilon}(s,y;\theta) -V^{\phi}_{t,\upsilon}(s,y;\theta)| \leq \hat{\gamma}\right)\notag\notag\\
     &\quad+ \nu_{\mathcal{S}}\left(\sup_{t\in\mathcal{T},y\in\mathcal{Y},\upsilon\in\Upsilon}| V^{\phi_t}_{t,\upsilon}(s,y;\theta) -V^{\phi}_{t,\upsilon}(s,y;\theta)| > \hat{\gamma}\right)\notag\\
     \leq & \nu_{\mathcal{S}}\left(\sup_{t\in\mathcal{T},y\in\mathcal{Y},\upsilon\in\Upsilon}| V_{t,\upsilon}(s,y;\theta) -V^{\phi_t}_{t,\upsilon}(s,y;\theta)| > \sum_{i=0}^{T-1}\gamma_{i}\right)\notag\\
      \leq & \sum_{t=0}^{T-1}\nu_{\mathcal{S}}\left(\sup_{y\in\mathcal{Y},\upsilon\in\Upsilon}| V_{t,\upsilon}(s,y;\theta) -V^{\phi_t}_{t,\upsilon}(s,y;\theta)| > \sum_{i=0}^{T-1}\gamma_{i}\right)\notag\\
      \leq & \sum_{t=0}^{T-1}\nu_{\mathcal{S}}\left(\sup_{y\in\mathcal{Y},\upsilon\in\Upsilon}| V_{t,\upsilon}(s,y;\theta) -V^{\phi_t}_{t,\upsilon}(s,y;\theta)| > \gamma_{t}\right)\notag\\
     <&  \sum_{t=0}^{T-1}\epsilon_t.\notag
\end{align}
By taking $\gamma_t=\frac{\epsilon^*}{2T},\hat{\gamma}=\frac{\epsilon^*}{2}$ and $\epsilon_t = \frac{\epsilon^*}{T}$ for all $t\in\mathcal{T},$ the conclusion holds.
\end{proof}

\begin{proof}[Proof of Lemma \ref{l0}]
As $\mathcal{S},\mathcal{Y},\Upsilon$ are compact and $f(\cdot,\cdot)$ is continuous,  it follows that $|\E [V_{0,\upsilon}(S_0,0;\theta)]|$ and $|\E [V^{\phi}_{0,\upsilon}(S_0,0;\theta)]|$ are bounded in $\upsilon\in\Upsilon$. Denote the upper bound by $B.$ As $h$ is continuous, it is uniformly continuous on $[-B,B]\times\Upsilon,$ i.e., for any $(V_1,\upsilon_1),(V_2,\upsilon_2)\in [-B,B]\times\Upsilon,$ $\gamma'>0,\gamma''>0,$ and we have
\begin{equation}
    |h(V_1,\upsilon_1)-h(V_2,\upsilon_2)|\leq\gamma'',\notag
\end{equation}
if 
\begin{equation}
    \sqrt{|V_1-V_2|^2+|\upsilon_1-\upsilon_2|^2} \leq \gamma'\notag
\end{equation}
holds. Taking $\upsilon_1=\upsilon_2=\upsilon$ yields that $V_1 = \E [V_{0,\upsilon}(S_0,0;\theta)]$ and $V_2=\E [V^{\phi}_{0,\upsilon}(S_0,0;\theta)]$. Hence, the conclusion follows. 
\end{proof}

\begin{proof}[Proof of Lemma \ref{l2}]
If $\mathscr{U}^*=\Upsilon,$ the conclusion is trivial. We only need to consider the case where $\mathscr{U}^*$ is a proper subset of $\Upsilon.$ We prove by contradiction and suppose that there exist $\gamma_0>0$ and $\left\{(\phi_n,\theta_n)\right\}_{n=1}^{\infty}$ such that
\begin{equation}
\label{limassumption}
    \lim_{n\to\infty} \underset{\upsilon \in \Upsilon}{\text{sup}} |h(\mathbb E [V^*_{0,\upsilon}(S_0,0)],\upsilon) - h(\mathbb E [V_{0,\upsilon}^{\phi_n}(S_0,0;\theta_n)],\upsilon)| = 0,
\end{equation}
and $d\left(\mathscr{U}^*,\upsilon^*({\theta_n},\phi_n)\right)>\gamma_0.$ We denote
\begin{equation}
    l_0 = \inf\left\{h(\mathbb E [V^*_{0,\upsilon}(S_0,0)],\upsilon):\upsilon\in\Upsilon\right\}\notag,
\end{equation}
and
\begin{equation}
    l_1 = \inf\left\{h(\mathbb E [V^*_{0,\upsilon}(S_0,0)],\upsilon):\upsilon\in\Upsilon,d\left(\mathscr{U}^*,\upsilon\right)>\gamma_0\right\}.\notag
\end{equation}
In view that $\upsilon\mapsto h(\mathbb E[ V^*_{0,\upsilon}(S_0,0)],\upsilon)$ is continuous on $\Upsilon,$ we have $l_1>l_0.$ Fix $\upsilon'\in \mathscr{U}^*.$ It holds that
\begin{align}
    \liminf_{n\to\infty}h(\mathbb E [V_{0,\upsilon'}^{\phi_n}(S_0,0;\theta_n)],\upsilon')&\geq \liminf_{n\to\infty}h(\mathbb E [V_{0,\upsilon^*({\theta_n},\phi_n)}^{\phi_n}(S_0,0;\theta_n)],\upsilon^*({\theta_n},\phi_n))\notag\\
    &\geq l_1>l_0=h(\mathbb E [V^*_{0,\upsilon'}(S_0,0)],\upsilon'),\notag
\end{align}
which leads to a contradiction to \eqref{limassumption}.\end{proof}

\begin{proof}[Proof of Theorem \ref{t2}]
Because $\mathcal{S},\mathcal{Y},\Upsilon$ are all compact and $f(\cdot,\cdot)$ is continuous, we have $|V_{0,\upsilon}(s_0,0;\theta)|$ and $|V^{\phi}_{0,\upsilon}(s_0,0;\theta)|$ are bounded for $(s_0,\upsilon)\in\mathcal{S}\times\Upsilon$. Denote the upper bound by $B.$ If the third condition holds, we have
\begin{align}
      \underset{\upsilon\in\Upsilon}{\text{sup}} |\mathbb E [V_{0,\upsilon}(S_0,0;\theta)]- \mathbb E[V^{\phi}_{0,\upsilon}(S_0,0;\theta)]| &\leq \mathbb E[\underset{\upsilon\in\Upsilon}{\text{sup}}| V_{0,\upsilon}(S_0,0;\theta) - V^{\phi}_{0,\upsilon}(S_0,0;\theta)|]\notag\\
       &=\int_{\mathcal{S}_1}\underset{\upsilon\in\Upsilon}{\text{sup}}| V_{0,\upsilon}(s_0,0;\theta) - V^{\phi}_{0,\upsilon}(s_0,0;\theta)|p(s_0)\d s_0 \notag\\
       & \quad+\int_{\mathcal{S}
       \backslash \mathcal{S}_1}\underset{\upsilon\in\Upsilon}{\text{sup}}| V_{0,\upsilon}(s_0,0;\theta) - V^{\phi}_{0,\upsilon}(s_0,0;\theta)|p(s_0)\d s_0\notag\\
       &\leq 2B\epsilon_1'+\gamma_1'.\notag
   \end{align}
For any $\gamma'>0,$ we take $\gamma_1' = \frac{\gamma'}{4},$ $\epsilon_1'=\frac{\gamma'}{4B}$ and have
\begin{align}
    \underset{\upsilon\in\Upsilon}{\text{sup}} |\mathbb E[V^*_{0,\upsilon}(S_0,0)] - \mathbb E[V^{\phi}_{0,\upsilon}(S_0,0;\theta)]|&\leq \underset{\upsilon\in\Upsilon}{\text{sup}} |\mathbb E [V_{0,\upsilon}(S_0,0;\theta)] - \mathbb E[V^{\phi}_{0,\upsilon}(S_0,0;\theta)]|\notag\\
    &\quad+ \sup_{\upsilon\in \mathcal{Y}}|\mathbb E [V_{0,\upsilon}^{\phi}(S_0,0;\theta)]-\mathbb E[V^*_{0,\upsilon}(S_0,0)]| \notag\\
    &< (2B\epsilon_1'+\gamma_1') + \gamma_1'\notag\\
    &=\gamma'.\notag
\end{align}
By Lemmas \ref{l0}-\ref{l2},  the desired conclusion holds.
\end{proof}

\begin{proof}[Proof of Theorem \ref{thm:am_converge}]
From \eqref{argmin}, it follows that $\E [V^{\pi_{n,\epsilon}}_{0,\upsilon_{n+1}}(S_0,0)]\leq\E [V^{\pi_{n,\epsilon}}_{0,\upsilon_{n}}(S_0,0)]$ holds for all $n\geq1$. Hence, we have 
\begin{align}
\label{chain_neq}h(\E [V^{\pi_{n,\epsilon}}_{0,\upsilon_{n+1}}(S_0,0)],\upsilon_{n+1})\leq &h(\E [V^{\pi_{n,\epsilon}}_{0,\upsilon_{n}}(S_0,0)],\upsilon_n)\notag\\
< &h(\E [V^{*}_{0,\upsilon_{n}}(S_0,0)],\upsilon_n)+\frac{\epsilon}{n^2}\notag\\
\leq &h(\E [V^{\pi_{n-1,\epsilon}}_{0,\upsilon_{n}}(S_0,0)],\upsilon_n)+\frac{\epsilon}{n^2}\notag\\
\leq &h(\E [V^{\pi_{n-1,\epsilon}}_{0,\upsilon_{n}}(S_0,0)],\upsilon_n)+\frac{\epsilon}{n-1}-\frac{\epsilon}{n},
\end{align}
    which implies \begin{equation}
        h(\E [V^{\pi_{n,\epsilon}}_{0,\upsilon_{n+1}}(S_0,0)],\upsilon_{n+1})+\frac{\epsilon}{n} < h(\E [V^{\pi_{n-1,\epsilon}}_{0,\upsilon_{n}}(S_0,0)],\upsilon_{n})+\frac{\epsilon}{n-1},\notag
    \end{equation}
    for all $n\geq 1.$ As a result, $\left\{A^{\epsilon}_n= h(\E [V^{\pi_{n,\epsilon}}_{0,\upsilon_{n+1}}(S_0,0)],\upsilon_{n+1})+\frac{\epsilon}{n}\right\}_{n=1}^{\infty}$ is a decreasing sequence with the lower bound $\min_{\upsilon\in\Upsilon}h(\E [V^{*}_{0,\upsilon}(S_0,0)],\upsilon)$, which converges to some $A^{\epsilon}\geq \min_{\upsilon\in\Upsilon}h(\E [V^{*}_{0,\upsilon}(S_0,0)],\upsilon)>-\infty,$ i.e., 
    \begin{equation}
        \lim_{n\to\infty}h(\E [V^{\pi_{n,\epsilon}}_{0,\upsilon_{n+1}}(S_0,0)],\upsilon_{n+1})=\lim_{n\to\infty}h(\E [V^{\pi_{n,\epsilon}}_{0,\upsilon_{n+1}}(S_0,0)],\upsilon_{n+1})+\frac{\epsilon}{n} = A^{\epsilon}.\notag
    \end{equation}
    Fix a $\gamma>0$. For any $N\geq 1,$ there exists some $\epsilon_N>0$ such that \begin{equation}
          h(\E [V^{\pi_{n,\epsilon}}_{0,\upsilon}(S_0,0)],\upsilon)\leq h(\E [V^{\pi}_{0,\upsilon}(S_0,0)],\upsilon)+\delta(\upsilon,\upsilon_n)+\frac{\gamma}{(n+1)^2},\notag
     \end{equation}
     for any $\upsilon\in\Upsilon,\pi\in\tilde{\Pi},0<\epsilon<\epsilon_N$ and $1\leq n\leq N.$ Thus, we have \begin{equation}
         \delta(\upsilon,\upsilon_{n+1})+h(\E [V^{\pi_{n,\epsilon}}_{0,\upsilon_{n+1}}(S_0,0)],\upsilon_{n+1})\leq h(\E [V^{\pi_{n,\epsilon}}_{0,\upsilon}(S_0,0)],\upsilon)\leq h(\E [V^{\pi}_{0,\upsilon}(S_0,0)],\upsilon)+\delta(\upsilon,\upsilon_n)+\frac{\gamma}{(n+1)^2},\notag
 \end{equation}
 which yields that
 \begin{equation}
 \label{recursiv_neq}
    h(\E [V^{\pi_{n,\epsilon}}_{0,\upsilon_{n+1}}(S_0,0)],\upsilon_{n+1})-h(\E [V^{\pi}_{0,\upsilon}(S_0,0)],\upsilon)\leq\delta(\upsilon,\upsilon_n)-\delta(\upsilon,\upsilon_{n+1})+\frac{\gamma}{(n+1)^2},
      \end{equation}
for any $\upsilon\in\Upsilon,\pi\in\tilde{\Pi},0<\epsilon<\epsilon_N$ and $1\leq n\leq N.$ 
Summing equation \eqref{recursiv_neq} from $n=1$ to $n = N$, we deduce that
\begin{align}
   \sum_{n=1}^{N}[ h(\E [V^{\pi_{n,\epsilon}}_{0,\upsilon_{n+1}}(S_0,0)],\upsilon_{n+1})-h(\E [V^{\pi}_{0,\upsilon}(S_0,0)],\upsilon)] \leq& \delta(\upsilon,\upsilon_1)-\delta(\upsilon,\upsilon_{N+1})+\sum_{n=1}^N\frac{1}{(n+1)^2}\gamma\notag\\ \leq& \delta(\upsilon,\upsilon_1)-\delta(\upsilon,\upsilon_{N+1})+\gamma.\notag
\end{align}
Sending $N$ to infinity, we get
\begin{equation}
    \limsup_{N\to\infty} \sum_{n=1}^{N} [h(\E [V^{\pi_{n,\epsilon}}_{0,\upsilon_{n+1}}(S_0,0)],\upsilon_{n+1})-h(\E [V^{\pi}_{0,\upsilon}(S_0,0)],\upsilon)]\leq \delta(\upsilon,\upsilon_1)+\gamma<\infty,\notag
\end{equation}
which implies 
\begin{equation}
    \limsup_{N\to\infty,\epsilon\to0} h(\E [V^{\pi_{N,\epsilon}}_{0,\upsilon_{N+1}}(S_0,0)],\upsilon_{N+1})=\limsup_{N\to\infty,\epsilon<\epsilon_N} h(\E [V^{\pi_{N,\epsilon}}_{0,\upsilon_{N+1}}(S_0,0)],\upsilon_{N+1})\leq h(\E [V^{\pi}_{0,\upsilon}(S_0,0)],\upsilon).\notag
\end{equation}
By the arbitrariness of $\upsilon$ and $\pi,$ we obtain that 
\begin{equation}
    \limsup_{n\to\infty,\epsilon\to 0 } h(\E [V^{\pi_{n,\epsilon}}_{0,\upsilon_{n+1}}(S_0,0)],\upsilon_{n+1})\leq h(\min_{\upsilon\in\Upsilon}\E [V^{*}_{0,\upsilon}(S_0,0)],\upsilon)\leq \liminf_{n\to\infty,\epsilon\to 0 }h(\E [V^{\pi_{n,\epsilon}}_{0,\upsilon_{n+1}}(S_0,0)],\upsilon_{n+1}),\notag
\end{equation}
which completes the proof.
\end{proof}

The proof of Proposition \ref{thm:limit} requires the following lemma. 
\begin{lemma}
\label{p_qa}
If $f(y,\upsilon) = \upsilon+\frac{1}{1-\alpha}(y-\upsilon)^+,$ for any $\pi\in\tilde{\Pi},$ we have 
\begin{equation}
     \E [V^{\pi}_{0,\upsilon}(S_0,0)]|_{\upsilon=q_{\alpha}(F_{\pi})} = \min_{\upsilon\in \mathbb{R}} \E [V^{\pi}_{0,\upsilon}(S_0,0)].\notag
\end{equation}
\end{lemma}

\begin{proof}[Proof of Lemma \ref{p_qa}]
  We have \begin{align}
       &\frac{\partial}{\partial \upsilon}\mathbb{E}_{\pi} \left.\left[\upsilon+\frac{1}{1-\alpha}\left(\sum_{k=0}^{T-1}c(S_k,A_k,S_{k+1})-\upsilon\right)^+\right|Y_0=0\right]\notag\\
       =&\mathbb{E}_{\pi}  \left.\left[1-\frac{1}{1-\alpha}\mathds{1}\left(\sum_{k=0}^{T-1}c(S_k,A_k,S_{k+1})-\upsilon>0\right)\right|Y_0=0\right]\notag\\
       =&1-\frac{1}{1-\alpha}\mathbb{P}_{\pi} \left.\left(\sum_{k=0}^{T-1}c\left(S_{k},A_k,S_{k+1}\right)> \upsilon\right|Y_0=0\right)\notag\\
       =& \frac{1}{1-\alpha}(F_{\pi}(\upsilon)-\alpha).\notag
\end{align}
Clearly, the function $\upsilon\mapsto \E[V^{\pi}_{0,\upsilon}(S_0,0)]$ attains its minimum at $F_{\pi}(\upsilon)=\alpha.$ 

\end{proof}

\begin{proof}[Proof of Proposition \ref{thm:limit}]
Taking $\delta(\upsilon_1,\upsilon_2)=\frac{\Delta}{2}(\upsilon_1-\upsilon_2)^2$, together with Theorem \ref{thm:am_converge} and Lemma \ref{p_qa}, leads to the conclusion.
\end{proof}

\ \\
\noindent
{\bf Acknowledgements}
Y. Liu acknowledges financial support from the National Natural Science Foundation of China (Grant No. 12401624), the Chinese University of Hong Kong (Shenzhen) research startup fund (Grant No. UDF01003336) and Shenzhen Science and Technology Program (Grant No. RCBS20231211090814028) and is partly supported by the Guangdong Provincial Key Laboratory of Mathematical Foundations for Artificial Intelligence (Grant No. 2023B1212010001). X. Yu is supported by the Hong Kong RGC General Research Fund (GRF) under grant no. 15211524, the Hong Kong Polytechnic University under grant no. P0045654 and the Research Centre for Quantitative Finance at the Hong Kong Polytechnic University under grant no. P0042708.

\small
\bibliographystyle{apalike}
\bibliography{ref}

\end{document}